\documentclass[11pt]{article}
\usepackage{float}
\usepackage[margin=1in]{geometry}  
\usepackage[T1]{fontenc}

\usepackage{fullpage}
\usepackage{xspace}
\usepackage[table]{xcolor}
\usepackage{multirow}
\usepackage{hhline} 
\usepackage{tablefootnote}
\usepackage[font=small]{caption}

\usepackage{graphicx} % Required for inserting images
\usepackage{amsmath}
\usepackage{amsfonts}
\usepackage{physics}
\usepackage[textsize=tiny]{todonotes}
\usepackage{hyperref}
\usepackage{amsthm}
\usepackage{cleveref}
\usepackage{comment}
\usepackage{amssymb}
\usepackage{cleveref}
\usepackage{physics}
\usepackage{colortbl}
\usepackage{array}
\usepackage{enumitem}
\usepackage{complexity}
\usepackage{thm-restate}
\usepackage{thmtools}
\usepackage{bbm}

\newcommand{\PARITY}{\texttt{PARITY}\xspace}
\newcommand{\FANOUT}{\texttt{FAN-OUT}\xspace}

\newcommand{\AND}{\texttt{AND}\xspace}
\newcommand{\OR}{\texttt{OR}\xspace}
\newcommand{\NOT}{\texttt{NOT}\xspace}

\newcommand{\TOFFOLI}{\texttt{TOFFOLI}\xspace}

\newcommand{\THRESHOLD}{\texttt{THRESHOLD}\xspace}
\newcommand{\EXACTONE}{\texttt{EXACT[1]}\xspace}

\newcommand{\EXACT}{\texttt{EXACT}}

\newcommand{\THRESH}{\THRESHOLD}
\newcommand{\THR}{\texttt{TH} \xspace}
\newcommand{\SWAP}{\texttt{SWAP}}

\newcommand{\ACZ}{\textsf{AC}$^0$\xspace}

\newcommand{\QACZ}{\textsf{QAC}$^0$\xspace}

\newcommand{\QACZf}{\textsf{QAC}$^0_f$\xspace}

\newcommand{\calF}{\mathcal{F}}

\newcommand{\calO}{\mathcal{O}}

\newcommand{\calS}{\mathcal{S}}

\renewcommand{\poly}{\text{poly}}

\newcommand{\CZ}{\ensuremath{\texttt{CZ}}\xspace}

\renewcommand{\l}{\ell}
\renewcommand{\Pr}{\mathop{\bf Pr\/}}

\renewcommand{\S}{\mathcal{S}}

\newcommand{\bn}{\mathsf{bin}}

\theoremstyle{plain} 

\newtheorem{lemma}{Lemma}

\newtheorem{corollary}{Corollary}

\newtheorem{claim}{Claim}

\theoremstyle{definition} 

\newtheorem{fact}{Fact}
\theoremstyle{remark}

\def\DRAFT{1}

\ifnum\DRAFT=1
\newcommand{\malvika}[1]{\textcolor{purple}{[\textbf{Malvika:} {#1}]}}
\newcommand{\fran}[1]{\textcolor{red}{[\textbf{Fran:} {#1}]}}
\fi

\ifnum\DRAFT=0
\newcommand{\malvika}[1]{}
\newcommand{\fran}[1]{}
\fi

\title{Constant-Depth Unitary Preparation of Dicke States}
\ifnum\DRAFT=1
\author{
  Malvika Raj Joshi\thanks{University of California at Berkeley. \ Email: \url{malvika@berkeley.edu}. \ Supported by NSF grant 2311733 and DOE grant DE-SC0024124.} \footnotemark[3]
 \and 
  Francisca Vasconcelos\thanks{University of California at Berkeley. \ Email: \url{francisca@berkeley.edu}. \ Supported by the U.S. Department of Energy, Office of Science, under Award No. DE-SC0024124 and the Paul \& Daisy Soros Fellowship for New Americans.}\thanks{Both authors contributed equally to this work. Ordering is alphabetical.}
}
\date{}
\fi
\ifnum\DRAFT=0
\author{
  Anonymous Authors.
}
\fi
\date{}

\usepackage{extra_commands}

\begin{document}
\maketitle

\begin{abstract}
    Dicke states serve as a critical resource in quantum physics, metrology, communication, and computation. Notably, they function as the primary quantum resource for Decoded Quantum Interferometry~\cite{jordan2025optimization}, a recent candidate algorithm for quantum advantage in combinatorial optimization. However, unitary preparation of these highly entangled states within the standard quantum circuit model is fundamentally limited to logarithmic depth. Furthermore, all known constant-depth protocols for exact preparation rely on measurement and adaptive feed-forward.
    
    In this work, we present the first unitary, constant-depth protocols for \emph{exact} Dicke state preparation. We overcome the logarithmic-depth barrier by moving beyond the standard circuit model and leveraging global interactions (native to architectures such as neutral atoms and trapped ions). For constant-weight Dicke states, we provide an explicit circuit, utilizing unbounded \CZ gates (i.e. within the \QACZ circuit class), with polynomial ancilla overhead. Within this circuit class, we also describe a \emph{constant}-ancillae circuit for the constant-error approximation of weight-1 Dicke states (i.e. $W$ states). Granted additional access to the quantum \FANOUT operation (i.e. upgrading to the \QACZf circuit class), we also achieve a constant-depth construction for arbitrary-weight Dicke states.
    
    The inherent challenges in preparing super-constant-weight Dicke states without access to \FANOUT suggest these states as a natural witness for a state-synthesis separation between \QACZ and \QACZf. Such a separation would resolve a long-standing open question in quantum complexity theory~\cite{moore1999qac0}. Experimentally, it would imply a computational hierarchy amongst quantum hardware architectures in the constant-depth regime: placing systems capable of global \FANOUT operations (e.g. trapped ions) above those restricted to global \CZ operations (e.g. neutral atoms), with both strictly outperforming architectures bound by local geometry (e.g. standard superconducting lattices).
\end{abstract}

% SHORT ABSTRACT
% Dicke states serve as a critical resource in quantum metrology, communication, and computation. However, unitary preparation of Dicke states is limited to logarithmic depth in standard circuit models and existing constant-depth protocols require measurement and feed-forward. In this work, we present the first unitary, constant-depth protocols for exact Dicke state preparation. We overcome the logarithmic-depth barrier by moving beyond the standard circuit model and leveraging global interactions (native to architectures such as neutral atoms and trapped ions). Specifically, utilizing unbounded \texttt{CZ} gates (i.e. within the \textsf{QAC}$^0$ circuit class), we offer circuits for exact computation of constant-weight Dicke states, using polynomial ancillae, and approximation of weight-1 Dicke states (i.e. $W$ states), using only constant ancillae. Granted additional access to the quantum \texttt{FAN-OUT} operation (i.e. upgrading to the \textsf{QAC}$_f^0$ circuit class), we also achieve \emph{exact} preparation of arbitrary-weight Dicke states, with polynomial ancillae. These protocols distinguish the constant-depth capabilities of quantum architectures based on connectivity and offer a novel path toward resolving a long-standing quantum complexity conjecture.

\thispagestyle{empty}
\newpage

{\small\tableofcontents}

\thispagestyle{empty}
\newpage
\pagenumbering{arabic} 

\section{Introduction} \label{sec:intro}
Distinguished by their permutation symmetry and robustness against particle loss, Dicke states are a cornerstone of multipartite entanglement, with applications in quantum physics, metrology, communication, and computation.  Mathematically, for $k \in [n-1]$, the $n$-qubit weight-$k$ Dicke state is simply defined as the uniform superposition state over all computational basis states $\ket{x}$, i.e. for $\x \in \{0,1\}^n$ of Hamming weight $|\x|=k$, 
\begin{align}
    \ket{D^n_k} = \frac{1}{\sqrt{\binom{n}{k}}} \sum_{\substack{\x \in \{0,1\}^n: \\ |\x| = k}} \ket{\x}.
\end{align}
Robert Dicke originally introduced Dicke states in 1954 to explain super-radiance phenomena in atom ensembles \cite{dicke1954state}. However, these states have since come to be used in a variety of contexts, including: quantum game theory \cite{ozedmir2007games}, quantum networking \cite{prevedel2009network, chiuri2012network}, quantum secret sharing \cite{guo2024secret}, Heisenberg-limited metrology \cite{toth2012metrology}, quantum error correction \cite{ouyang2014codes,ouyang2021permutation}, quantum storage/memory \cite{ouyang2021storage, fleischhauer2002memory}, adiabatic optimization \cite{childs2002adiabatic}, variational quantum optimization \cite{hadfield2019qaoa, wang2020xy, cook2020qaoa, bartschi2020qaoa,golden2021qaoa}, radiative control \cite{bienaime2012control}, quantum transport \cite{rebentrost2009transport}, and entanglement benchmarking \cite{somma2006lb}. Notably, Dicke states serve as the primary quantum resource in Decoded Quantum Interferometry (DQI) \cite{jordan2025optimization}---a promising new algorithmic framework for achieving quantum advantage in constraint satisfaction.

Current strategies for preparing Dicke states with quantum computers face a fundamental trade-off between circuit-depth and control-complexity. Standard quantum compilation strategies are typically restricted to constant-width gates, which (even for architectures with all-to-all connectivity) imposes an entanglement-distribution bottleneck. This results in circuit depths that scale at least logarithmically, i.e. $\Omega(\log n)$, in the number of qubits $n$ \cite{yuan2025depth}. 
This scaling introduces a `decoherence wall': as $n$ increases, the circuit duration eventually exceeds the coherence window of the device, causing state fidelity to vanish exponentially for large system sizes \cite{Preskill2018quantumcomputingin}. To bypass this depth-barrier, recent proposals have turned to adaptive quantum circuits, which achieve constant-depth using mid-circuit measurements and classical feed-forward \cite{Buhrman2024statepreparation,yu2025efficient}. While theoretically efficient, these protocols require fast, low-latency classical logic to process measurement outcomes within device coherence times. This introduces notable control complexity and non-unitary noise channels, which are challenging to mitigate.

In this work, we crucially demonstrate that this depth-adaptivity trade-off is \emph{unnecessary} for Dicke state synthesis. Concretely, \textbf{we propose the first simultaneously \emph{unitary} and \emph{constant-depth} quantum circuit constructions for exact preparation of arbitrary-weight $n$-qubit Dicke states.} 

From a many-body physics perspective, the measurement-free, constant-depth generation of the long-range, multipartite entanglement characteristic of Dicke states might appear to contradict fundamental speed limits. In systems restricted to local interactions, the propagation of information is governed by Lieb-Robinson bounds \cite{lieb1972finite}, which dictate that global entanglement generation across a diameter $L$ requires time $t \propto L$ \cite{bravyi2006lieb}. However, our approach bypasses this constraint by leveraging architectures with \emph{global} interactions. Notably, the strict linear causal bounds of local systems vanish in the presence of long-range or all-to-all couplings \cite{tran2020hierarchy}. By exploiting these non-local resources, we demonstrate that complex, permutation-invariant entanglement can be generated in constant quantum time, consistent with the physics of long-range interacting systems.

From a practical perspective, the global interactions used in our protocols are not merely theoretical wishes, but an increasingly viable reality for a number of leading experimental quantum architectures. In trapped ion systems, many-qubit \FANOUT gates are realizable via global Mølmer–Sørensen (MS) interactions, which generate GHZ-like correlations without decomposition into pairwise gates \cite{monz2011entangle, figgatt2019parallel}. In neutral atom arrays, many-qubit \CZ gates can be implemented via a Rydberg blockade mechanisms \cite{isenhower2011multibit,levine2019parallel, delakouras2025multi}. In fact, a recent breakthrough leveraged reconfigurable atom arrays to demonstrate transversal logical multi-qubit \CZ gates and apply constant-depth multi-body logic across 48 logical qubits simultaneously \cite{bluvstein2024logical}.

Finally, from the perspective of quantum complexity theory, our protocols are explicitly implemented in the \QACZ and \QACZf circuit classes. Note that \QACZ (\QACZf) is simply the family of constant-depth, polynomial-sized circuits comprised of arbitrary single-qubit and global \CZ (\FANOUT) gates. Although \QACZf trivially contains \FANOUT, it remains unknown whether \QACZ can implement \FANOUT. Since the \FANOUT gate can straightforwardly be used to compute the GHZ/CAT state, most of the literature on \QACZ versus \QACZf has focused on synthesis of GHZ states and its variants \cite{rosenthal2021qac0,anshu2025computational,joshi2025improvedlowerboundsqac0}. However, in this work, by instead focusing on the ability of \QACZ~and \QACZf~to a prepare a different class of highly entangled states, we achieve interesting new complexity implications. For example, by the recent lower-bounds of \cite{parham2025quantumcircuitlowerbounds}, our exact implementation of the $W$ state in \QACZ~implies that \QACZ~is not in the first level of the magic hierarchy. Furthermore, this work achieves protocols for preparation of \emph{arbitrary}-weight Dicke states in \QACZf, but only \emph{constant}-weight Dicke states in \QACZ. This implies that a lower-bound against preparation of any $\omega(1)$-weight Dicke state in \QACZ suffices to achieve a state synthesis separation between \QACZ and \QACZf --- a novel path towards resolving a major long-standing open question in quantum complexity theory \cite{moore1999qac0}.

\paragraph{Paper Overview.} The remainder of the paper will introduce necessary background, give a conceptual overview of the key results, and offer a comparison to related prior work. Concretely, \Cref{sec:prelims} provides an overview of the necessary background and results from quantum complexity theory used in this work. \Cref{sec:qac0_constant_weight} describes our first main series of results: exact and approximate unitary preparation of constant-weight Dicke states via global \CZ gates. First, we establish a direct connection between computation of the \EXACT$_k$ Boolean function and preparation of weight-$k$ Dicke states ($\ket{D^n_k}$) in \QACZ, for constant $k=\calO(1)$. We then offer an explicit polynomial-sized \QACZ circuit for exact computation of \EXACT$_k$ (for $k=\calO(1)$), resulting in a polynomial-sized \QACZ circuit for exact preparation of $\ket{D^n_k}$. We also offer an explicit constant-sized \QACZ circuit for approximation of \EXACT$_1$, resulting in a \QACZ circuit for constant-error approximation of the $W$ state, with only constant ancillae. 
\Cref{sec:qac0f_aritrary_weight} discusses our other main result: exact unitary preparation of arbitrary-weight Dicke states via global \FANOUT gates. In other words, we offer an explicit polynomial-sized \QACZf circuit for exact computation of arbitrary-weight Dicke states.  \Cref{sec:comp_prior_work} compares our unitary, constant-depth Dicke state constructions to prior unitary and measurement-based constructions for Dicke state preparation. Finally, \Cref{sec:discussion} offers a discussion of the results, including directions for future work.

\section{Preliminaries} \label{sec:prelims}
Before describing the results of the paper in more detail, we begin with an overview of relevant prior work in quantum circuit complexity, specifically pertaining to the \QACZ and \QACZf circuit classes.

\paragraph{The \QACZ Circuit Class.} The class \QACZ is defined as the family of constant-depth, polynomial-sized quantum circuits composed of arbitrary single-qubit gates and global (unbounded) \CZ gates. The global \CZ gate on a subset of qubits $S \subseteq [n]$ applies a phase of $-1$ if and only if all qubits in $S$ are in the state $\ket{1}$, i.e. for $\x \in \bin^n$,
\begin{align}
    \CZ_S \ket{\x} = \begin{cases}
        -\ket{\x}, & \text{if } \x = 1^{|S|} \\
        \ket{\x}, & \text{otherwise}
    \end{cases}.
\end{align}
Note that in this model, arbitrary single-qubit unitaries are considered depth-0 operations.

\QACZ was introduced by~\cite{moore1999qac0} as the natural quantum analogue of the classical circuit class \ACZ---which is defined by unbounded \AND, \OR, and \NOT gates. \QACZ mirrors this structure via the unbounded \TOFFOLI gate, which, for a given set $S$, flips a target qubit $t$ iff $\bigwedge_{i \in S} x_i = 1$. Note that the unbounded \TOFFOLI can be implemented in depth-0 \QACZ by conjugating a global \CZ gate with Hadamard gates on the target $t$. By DeMorgan's laws, the unbounded \OR gate is similarly implementable in depth-1. 

A structural characterization of \QACZ, particularly relevant to physical implementations, is the ``normal form'' due to \cite{rosenthal2021qac0}. Rosenthal showed that any \QACZ circuit can be manipulated into a form consisting of a single layer of single-qubit rotations, followed by a constant number of layers of product state reflections. A product state reflection is a unitary of the form $R_S(\phi) = I - 2\ketbra{\phi}{\phi}$, where $\ket{\phi} = \bigotimes_{i \in S} \ket{\phi_i}$ is a product state, composed of single-qubit states $\ket{\phi_i}$. Conceptually, this frames \QACZ as the class of circuits composed of global, but low-entangling (Schmidt-rank 2) gates.

\paragraph{The \QACZf Circuit Class.} While classical \ACZ allows for unbounded \FANOUT gates, Moore's original definition of \QACZ did not. The class \QACZf is defined as \QACZ imbued with the unbounded quantum \FANOUT gate. The quantum \FANOUT gate copies a control bit $c$ into a register of target qubits $t_1, \dots, t_m$, i.e.
\begin{align}
    \FANOUT \ket{c}\ket{t_1 \dots t_m} = \ket{c}\ket{t_1 \oplus c \dots t_m \oplus c}.
\end{align}

A long-standing open question in quantum complexity theory, originally posed by \cite{moore1999qac0}, is whether \QACZ $=$ \QACZf~? This is equivalent to asking if global \CZ gates can simulate \FANOUT in constant depth. The converse is known to be true, i.e. \cite{takahashi2012collapse} proved that unbounded \FANOUT (combined with single-qubit gates) can simulate unbounded \CZ in constant depth. The power of \QACZf is evidenced by its ability to compute the \PARITY function. While \PARITY $\notin$ \ACZ~\cite{hastad1986switch}, \PARITY $\in$ \QACZf since \PARITY can be implemented via Hadamard conjugation of \FANOUT \cite{moore1999qac0}.

Beyond \PARITY, \QACZf is capable of computing complex arithmetic functions, including phase estimation and the quantum Fourier transform~\cite{hoyer2005fanout}. Crucial to the results of this work is the ability of \QACZf to exactly compute the classical \THRESH Boolean function.
\begin{fact}[Exact \THRESH in \QACZf~\cite{takahashi2012collapse}] \label{fact:thresh_qac0f}
    For any $k \in [0, n]$ and all $\x \in \bin^n$, the function 
    \begin{align}
        \textnormal{\THRESH}_k(\x) = \begin{cases}
            1, & \text{if } |\x|\leq k\\
            0, & \text{otherwise}
        \end{cases}
    \end{align}
    can be implemented exactly by a \QACZf, with polynomial ancillae.
\end{fact}

\paragraph{Progress on \QACZ vs \QACZf.}
Nearly 30 years since Moore's original proposal \cite{moore1999qac0}, it still remains open as to whether \QACZ = \QACZf.
However, there has been notable progress on proving both lower- and upper-bounds in recent years. On the end of lower-bounds, building upon prior results by \cite{fang2006qlb,pade2020depth2qaccircuitssimulate,rosenthal2021qac0,nadimpalli2024pauli,fenner2025tightboundsdepth2qaccircuits}, \cite{anshu2025computational} proved a slightly super-linear ancillae lower-bound for arbitrary constant-depth \QACZ and \cite{joshi2025improvedlowerboundsqac0} established an unlimited ancillae lower bound for depth-3 \QACZ. Excitingly, structural properties underlying the lower-bound of \cite{nadimpalli2024pauli} have resulted in time-efficient learning of \QACZ unitaries~\cite{vasconcelos2025learning}.

More pertinent to this work, however, are recent upper-bounds on \QACZ. Crucially, it is established that \QACZ can simulate \FANOUT (and thus \QACZf), with access to super-polynomial ancillae. \cite{rosenthal2021qac0} first provided an \emph{approximate} implementation, utilizing exponential ancillae. \cite{foxman2025pru} observed that by restricting this \FANOUT gadget to $\mathcal{O}(\log n)$ qubits the ancilla-overhead is reduced to polynomial, which they leveraged to implement weak pseudorandom unitaries in \QACZ. Most recently, \cite{grier2026qac0} employed amplitude amplification to promote Rosenthal's approximate implementation to an \emph{exact} one.
\begin{fact}[Exact \FANOUT in Exp-Size \QACZ \cite{grier2026qac0}] \label{fact:grier_exact}
    The unbounded quantum \FANOUT operation on $n$ qubits can be implemented exactly by a \QACZ circuit of constant depth and size $\calO(2^n)$.
\end{fact}

\section{Constant-Weight Dicke States via Global \CZ Gates} \label{sec:qac0_constant_weight}
We will now describe how to unitarily prepare constant-weight Dicke states in constant-depth, via global \CZ operations (i.e. within the \QACZ circuit class). Our first main result is an explicit \QACZ circuit for exact preparation of constant-weight Dicke states, requiring polynomial ancillae.
\begin{restatable}[Exact Constant-Weight Dicke in \QACZ]{theorem}{exactdickeqaczero} \label{thm:exact_dicke}
    For any $k=\calO(1)$, there exists a \QACZ circuit for exact preparation of the weight-$k$ Dicke state, using $\calO(n^{k+1})$ ancillae.
\end{restatable}
\noindent As a direct consequence of the lower-bounds of \cite{parham2025quantumcircuitlowerbounds}, we note that \Cref{thm:exact_dicke} implies that \QACZ is not in the first level of the magic hierarchy.
\begin{restatable}{corollary}{qaczmagic} \label{thm:qacz_magic}
    \textnormal{\QACZ$\not\subseteq$ \textsf{MH$_1$}} 
\end{restatable}
\noindent Our second main result is an explicit circuit for constant-error approximation of the $W$ state (weight-1 Dicke state), using only \emph{constant} ancillae.
\begin{restatable}[Approximate $W$ State in \QACZ]{theorem}{approxw}\label{thm:wn_qac}
    For any constant error $\varepsilon \in (0,1)$ there exists a
    \QACZ~circuit family $\{C_n\}_{n\in\mathbb{N}}$ with $a=O(1)$ ancilla qubits such that
    for all $n$,
    \begin{align}
        \calF\bigl( C_n \ket{0^{n}}\ket{0^a},\, \ket{W_n}\ket{0^{a-1}} \bigr)
      \ge 1 - \varepsilon,
    \end{align}
    where $\calF(\cdot,\cdot)$ denotes the standard state fidelity.
\end{restatable}
\noindent To achieve these results, we rely on the \EXACT$_k$ Boolean function (closely related to the previously described $\THRESH_k$ function) which, for $\x \in \bin^n$, is defined as
\begin{align} \label{eqn:exact_defn}
    \EXACT_k(\x) = \begin{cases}
        1, & \text{if } |\x|=k \\
        0, & \text{otherwise} 
    \end{cases}.
\end{align}
Notably, we make three primary contributions: 
\begin{enumerate}
	\item For every $k=\calO(1)$, we offer a general \QACZ reduction from computation of the \EXACT$_k$ Boolean function to preparation of the weight-$k$ Dicke state, with no additional ancilla overhead. In other words, we demonstrate an equivalence between computation of constant-Hamming weight \EXACT~functions and preparation of constant-weight Dicke states in \QACZ.
	\item For every $k=\calO(1)$, we offer an \emph{exact} $\calO(n^{k+1})$-ancillae \QACZ implementation of \EXACT$_k$.
    \item We offer a \emph{constant}-ancillae \QACZ circuit for constant-error approximation of \EXACT$_1$.
\end{enumerate}
Note that the proofs of \Cref{thm:exact_dicke} and \Cref{thm:wn_qac} follow straightforwardly from, respectively, plugging 
the exact and approximate \EXACT~implementations of (2) and (3) into the reduction of (1). Thus, the remainder of this section will focus on describing the constructions for (1), (2), and (3) as well as challenges in extending beyond constant weight.

\subsection{A \texorpdfstring{\QACZ}{Lg} Reduction from \texorpdfstring{$\ket{D^n_k}$}{Lg} to \EXACT\texorpdfstring{$_k$}{Lg}} \label{sec:reduction_overview}
We begin by establishing the formal connection between preparation of weight-$k$ Dicke states and computation of the \EXACT$_k$ Boolean function (full proof in  \Cref{sec:exact_dicke_reduct}). In particular, we show that preparation of the weight-$k$ Dicke state can be reduced to computation of the \EXACT$_k$ function. We will, thus, refer to this result as the Dicke-to-\EXACT~reduction.

To begin, we assume access to phase oracle oracle for computation of the \EXACT$_k$ Boolean function, i.e. a unitary $U_k$ that for any basis string $\x \in \bin^n$ performs the mapping
\begin{align}
    U_k \ket{\x} = \begin{cases}
        -\ket{\x}, & \text{if } |\x| = k \\
        \ket{\x}, & \text{otherwise}
    \end{cases}.
\end{align}
Note that such a phase oracle can be implemented from a \QACZ circuit that computes \EXACT$_k$ into an output bit simply by running the circuit, applying a $Z$ gate to the output, and running the circuit again in reverse (to uncompute the output).

Now, using this \QACZ phase oracle implementation of \EXACT$_k$ we show that, for any $k=\calO(1)$, the weight-$k$ Dicke state $\ket{D^n_k}$ can be computed with no additional ancilla overhead. The construction proceeds in two stages: 1) initialization and 2) amplification. 

First, we prepare a parameterized product state $\ket{\veta_\theta} = R_y(2\theta)^{\otimes n} \ket{0^n}$, using a single layer of rotation gates. The overlap probability between this state and the desired Dicke state, $p^n_k(\theta) = |\braket{\veta_\theta | D^n_k}|^2$, is binomially distributed. The maximum of this distribution is lower-bounded by $e^{-k}$, which is constant in the constant-weight regime $k=\calO(1)$.

Second, we employ exact amplitude amplification \cite{grover1996amp,brassard2002quantum, grier2026qac0} to boost this overlap to exactly $p^n_k(\theta) =1$, thereby achieving an exact implementation of $\ket{D^n_k}$. Note, for most initial success probabilities, standard amplitude amplification will oscillate around probability $1$ rather than converging exactly to it. To circumvent this, we invoke the Intermediate Value Theorem, showing that the rotation angle $\theta$ can be tuned such that the initial overlap $p^n_k(\theta)$ perfectly satisfies the fixed-point condition for a specific integer number of Grover iterations, $\ell$. Crucially, since the maximum achievable overlap is constant in $n$, the required number of iterations $\ell$ is also constant. Thus, the \QACZ circuit simply runs $\ell=\calO(1)$ iterations of amplitude amplification, in each iteration applying a reflection about the weight-$k$ Hamming subspace (i.e. the  \EXACT$_k$ phases oracle) and a reflection about the initial state $\ket{\veta_\theta}$. Note that the reflection about $\ket{\veta_\theta}$ can be implemented as a depth-1 \QACZ circuit, since $\ket{\veta_\theta}$ is a product state.

\subsection{Implementing Constant-Weight \EXACT~in \texorpdfstring{\QACZ}{Lg}}
We will now offer exact (polynomial-ancillae) constant-weight \EXACT~and approximate (constant-ancillae) \EXACT$_1$ implementations in \QACZ, which can be used to achieve the desired Dicke states of \Cref{thm:exact_dicke} and \Cref{thm:wn_qac}, via the the previously described Dicke-to-\EXACT~reduction. 

Our implementations will crucially leverage gates described in \Cref{sec:prelims}. Specifically, \AND and \OR gates can be implemented as depth-1 \QACZ circuits and the \NOT gate can be implemented in depth-0. Furthermore, \cite{grier2026qac0} showed that up to $\calO(\log n)$-\FANOUT can be performed exactly in \QACZ (\Cref{fact:grier_exact}).

\paragraph{Exact \EXACT$_k$.}  Rather than direct implementation of the \EXACT$_k$ Boolean function, we instead offer an explicit implementation of $\THRESH_k$ in \QACZ with $\calO(n^{k+1})$ ancillae (full proof in \Cref{sec:const_exact_proof}). By implementing $\THRESH_k$ for all constant $k$, \EXACT$_k$ can be computed from the $\THRESH_k$ and $\THRESH_{k-1}$ implementations via the following relation.
\begin{restatable}{fact}{exactthreshdef} \label{fact:exact_thresh}
    \EXACT$_k(\x)$ = $\THRESH_k(\x)$ $\land$ $\neg$ $\THRESH_{(k-1)}(\x)$
\end{restatable}
\noindent Note that the depth of this circuit in \QACZ is only that of the $\THRESH_k$ gate and two additional layers. The ancilla count is also asymptotically equivalent to that of the $\THRESH_k$ implementation.

Our exact \QACZ implementation of $\THRESH_k$ is achieved via a circuit implementing the following recursive formula (denote $\THRESH_k=\THR_k$):
\begin{align}
    \THR_{k}(\x) = \andl_{i \in [\log n]} \lr{ \
     \THR_0(\x_{\S_{i,0}}) \lor \THR_0(\x_{\S_{i,1}}) \ 
    \orl_{k' \in [k-1]} \lr{\THR_{k'}(\x_{\S_{i,0}}) \land \THR_{k-k'}(\x_{\S_{i,1}}) }},
\end{align}
where, for $b \in \{0,1\}$, $\S_{i,b} \subseteq [n]$ is the set of numbers $j \in [n]$, such that the $i^\text{th}$ bit of their binary representation is $b$, i.e. $\S_{i,b} = \{j \in [n]: \bn(j)_i = b\}$. Note that the \QACZ circuit crucially relies on the ability to perform $\calO(\log n)$-\FANOUT to parallelize the first \AND operation over the $\log n$ binary indices $i$. The correctness of this decomposition and $\calO(n^{k+1})$ ancillae count is proven via a dynamic programming analysis, for which we refer the reader to \Cref{sec:const_exact_proof}.

\paragraph{Approximate \EXACT$_1$.} Running our previous exact implementation to compute \EXACT$_1$ would require $\calO(n^2)$ ancillae. Instead, for any constant error $\varepsilon \in (0,1)$, we offer a randomized constant-ancillae circuit the computes \EXACT$_1$ correctly with probability at least $1-\varepsilon$.

To do so, we first construct a simplified gadget $G_\calS$ that, for input string $\x \in \{0,1\}^n$, checks if all the bits in some random subset $\calS \subset [n]$ are of value 0 and if at least one bit in $\bar{\calS}$ has value 1. Mathematically, the gadget can be expressed as the Boolean formula:
\begin{align}
    G_\calS(\x) = \left(\bigwedge_{i \in \calS} \neg~x_i\right) \land \left(\bigvee_{j \in \bar{\calS}} x_j\right). 
\end{align}
We then define the distribution $\mathfrak{S}$ over subsets $\calS \subset [n]$, where for each bit in $i \in [n]$ the bit $i$ is included in subset $\calS$ with probability $1/2$ (otherwise it is in the complementary subset $\bar{\calS}$). According to this distribution, the probability that the gadget outputs 1 is
\begin{align}
    \Pr_{\calS \sim \mathfrak{S}}\:[~G_{\calS}(\x)=1~] = 
    \begin{cases}
        0, & \text{if } |\x| = 0 \\
        1/2^{|\x|}, & \text{if } |\x| \geq 1
    \end{cases}.
\end{align}
Therefore, $G_{\calS}(\x)$ outputs $1$ correctly for $|x|=1$ with probability $1/2$ and incorrectly (for $|x|\geq 2$) with probability at most $1/4$ (note that for $|x|=0$ it is always correct). 

To widen the probability gap between the correct classification for $|x|=1$ and incorrect classification for $|x| \geq 2$, the circuit will run the gadget $G_{\calS_i}$ a constant number of times, $t$, instantiated with an independently randomly chosen sub-subset $\calS_i \sim \mathfrak{S}$ for each $i \in [t]$. For each run of the gadget, it stores the value of $G_{\calS_i}(\x)$ in a unique ancilla register. After all $t$ runs, it then computes a \THRESH over the $t=\calO(1)$ output ancillae, specifically checking whether the total number of successes, which we denote $N^{(t)}_{x}$, is at least $3t/8$ (indicating that at least $3/8$ of the $t$ runs resulted in $G_{\calS_i}(\x)=1$). Crucially, note that since this threshold is only being computed over $t=\calO(1)$ bits (not $n$), it can be computed with a constant number of ancillae.

Analyzing now over the distribution of the $t$ independently chosen $\calS_i$ subsets, i.e. $\mathbf{S}=\{\calS_i\}_{i \in [t]}\sim \mathfrak{S}^{\times t}$, Chernoff bounds are used to shows that probability that the gadget output is incorrect for \text{any} input string $\x \in \{0,1\}^n$ is upper-bounded as,
\begin{align}
    \Pr_{\mathbf{S}\sim \mathfrak{S}^{\times t}}\left(N^{(t)}_{|\x|=1} \leq 3t/8 \text{ or } N^{(t)}_{|\x|=1} \leq 3t/8\right) \leq \exp\left(-t/64\right).
\end{align}
Therefore, to ensure the error of an incorrect output is at most the constant error $\varepsilon$, we must set $t \geq \lceil 64 \log(1/\varepsilon) \rceil = \calO(1)$.

Note that for any specific choice of $\mathbf{S}$, the gadget will incorrectly compute \EXACT$_1(\x)$ for several inputs $\x  \in \bin^n$. However, the construction can be derandomized via the error analysis of the full Dicke state prepartion circuit, i.e. when this construction is plugged into the full \EXACT-to-Dicke reduction circuit of the last section. In particular, in \Cref{sec:approx_w_analysis}, we show that there exists a specific choice of the subsets $\textbf{S}^* \in \mathfrak{S}$, such that the fidelity between the states output by the \EXACT-to-Dicke reduction for the exact \EXACT$_1$ implementation and this approximate \EXACT$_1$ implementation (instantiated with subsets $\textbf{S}^*$) is at least $1-\varepsilon$.

\subsection{Challenges in Extending Beyond Constant-Weight}
A natural question is whether it is possible to implement Dicke states beyond constant-weight in \QACZ. Our proposed approach faces two key barriers to achieving this. First, our exact \EXACT$_k$ implementation cannot be extended to $k=\omega(1)$ without the ability to perform $\calO(n)$-\FANOUT. Second, for $\lambda \in (0,1)$, the Binomial probability of exactly $k=\lambda n$ successes is
\begin{align} \label{eqn:binom}
    p^n_k(\lambda) = {n \choose k} \: \lambda^k \: (1-\lambda)^{n-k} \approx \frac{1}{\sqrt{2\pi n \lambda (1-\lambda)}},
\end{align}
which for $0 \ll \lambda \ll 1$ is $\Theta (n^{-1/2})$. This implies that even given access to arbitrary-weight \EXACT~in \QACZ, in the regime where $0 \ll k \ll n$, the amplitude on the superposition over weight-$k$ strings in the \EXACT-to-Dicke reduction is not constant. Thus, a constant number of rounds of amplitude amplification is not sufficient to boost the success probability to 1. Since the number of rounds of amplitude amplification directly corresponds to the depth of the circuit, this implies the approach fundamentally cannot work in \QACZ for super-constant weights.

\section{Arbitrary-Weight Dicke States via Global \FANOUT} \label{sec:qac0f_aritrary_weight}
We will now show that, with additional access to the global \FANOUT gate, both of the previously described obstacles for implementing super-constant-weight Dicke states in \QACZ can be surpassed, enabling exact preparation of arbitrary-weight Dicke states in \QACZf. Specifically, we prove the following theorem.
\begin{restatable}[Exact Arbitrary-Weight Dicke in \QACZf]{theorem}{exactqaczf} \label{thm:exact_qac0f}
    For sufficiently large $n$, $k\in [n]$, there exists a \QACZf circuit, using $a=O(n^2 \sqrt{\log n})$ ancillae, that maps the input state $\ket{0^{n+a}}$ to the exact, clean Dicke state $\ket{D^n_k}\ket{0^a}$.
\end{restatable}
The first barrier, i.e. exact implementation of \EXACT$_k$ for $k=\omega(1)$, is straightforwardly overcome in \QACZf by leveraging prior work. In particular, as discussed in \Cref{sec:prelims}, \cite{takahashi2012collapse} showed that $\THRESH_k$ can be implemented in \QACZf for arbitrary weight $k \in [0,n]$ (\Cref{fact:thresh_qac0f}). By \Cref{fact:exact_thresh}, this can therefore be used to exactly implement $\EXACT_k$ in \QACZf for arbitrary weight $k \in [0,n]$. Thus, the key contribution of this \QACZf construction is in overcoming the (previously described) amplitude amplification barrier.

Namely, for $0 \ll k \ll n$, the success probability of a single Dicke-to-\EXACT~reduction scales as $\Theta(1/\sqrt{n})$. Constant-depth amplitude amplification requires a constant initial overlap, which implies that a single instance is insufficient in constant-depth. To resolve this, we employ a parallel amplification strategy, inspired by \cite{rosenthal2021qac0}'s sampling-based procedure for approximate GHZ state preparation in \QACZ. Instead of amplifying a single state with small (non-constant) overlap with the desired Dicke state (requiring super-constant depth), we generate a polynomial number of independent attempts in parallel and coherently extract the ``successful'' result. 

At the highest level, the protocol implements a form of parallelized amplitude amplification. First, a product of independent states (each with small overlap with the desired Dicke state) is amplified into a global superposition for which each superposition branch contains exactly one copy of the desired Dicke state. Then, we propose a gadget which swaps the singular good Dicke state of each superposition branch into an output register. Finally, we show that the ancillae can be uncomputed to obtain the desired $\ket{D^n_k}\ket{0^a}$ state. We now offer a more detailed description of the four main steps of the procedure (a formal proof can be found in \Cref{sec:qaczf_arb_dicke}).

\iffalse
\paragraph{1) Initalization.}We initialize $M$ parallel registers (blocks) $T_1, \dots, T_M$, each of size $n$. On each block, we apply independent rotations to create a biased superposition (i.e. preparing the state $\ket{\veta_\theta}$ as in \Cref{sec:reduction_overview}). Crucially, while the physical state of each block is a product state, we can mathematically decompose it into a ``target'' component (lying in the Hamming-weight-$k$ subspace) and a ``junk'' component (orthogonal to the target).
Explicitly, the state of the $i^{th}$ block can be written as:
\malvika{We're missing the $a$ register here, I'm adding it}
\begin{equation}\label{eq:start}
\ket{\psi}_{T_i} = \sqrt{p_{\good}}  \ket{D^n_k}_{T_i} \ket{1}_{a_i} + \sqrt{1-p_{\good}} \ket{\text{junk}}_{T_i} \ket{0}_{a_i},
\end{equation}
where $p_{\good}$ is the small overlap probability ($p_{\good} \approx 1/\sqrt{n}$) and $\ket{\text{junk}}$ is a normalized superposition of basis states with weight $\neq k$.
The global system state is the tensor product of these $M$ blocks, i.e. $\ket{\varphi_0} =    \bigotimes_{i \in [M]}  \ket{\psi}_{T_i}$. By expanding this product, we can view the system as a superposition of $2^M$ branches, indexed by a string $\x \in \{0,1\}^M$ which records whether each block ``succeeded'' ($\x_i=1$) or ``failed'' ($\x_i=0$), i.e.
\begin{align}
    \ket{\varphi_0} &=  \lr{\sum_{\x \in \bin^M} \sqrt{p({\x})} \cdot \ket{\psi_0(\x)}_{T_1 \dots T_m} \ket{\x}_{A}} 
\end{align}
where $p(\x) = \lr{p_{\good}}^{|\x|} \lr{1-p_{\good}}^{1-|\x|}$ and
\begin{align}
    \ket{\psi_0(\x)} = \lr{\bigotimes_{i \in [M] : x_i = 1} \ket{D^n_k}_{T_i} } \tens \lr{ \bigotimes_{i \in [M] : x_i = 0} \ket{\junk}_{T_i}} 
\end{align}
The goal of the subsequent circuit is to coherently isolate the ``good'' subspace (i.e. the branches where $\x \neq 0^M$) and extract a single copy of $\ket{D^n_k}$.

\paragraph{2) LSB Selection.}To extract a single copy of the Dicke state, we must select a unique block index $i^*$ from the set of successful blocks. We implement a coherent priority encoder that selects the \emph{Least Significant Bit} (LSB) of the flag string $\x$. Specifically, we implement a one-hot selection register $S$ using constant-depth classical logic:
\begin{equation}
\ket{\oh(\lsb(\x))}_S = \ket{s_1 \dots s_M}_S \quad \text{where } s_j = \x_j \land \neg \left( \bigvee_{k < j} \x_k \right).
\end{equation}
This ensures that for every valid branch ($\x \neq 0^M$), there is exactly one bit of $S$ set to $1$, i.e. $s_{i^*}=1$  corresponding to the first successful block $\x_{i^*}$.

\paragraph{3) Extraction via Controlled-SWAP.}We introduce a final target register $Q$, initialized to $\ket{0^n}$. We then perform a parallel controlled-\texttt{SWAP} operation: for each $j \in [M]$, we swap the contents of block $T_j$ and target $Q$ conditioned on the selection bit $s_j$, i.e. $\prod_{j=1}^M \texttt{C}_{s_j}\texttt{SWAP}(T_j, Q)$. 
Crucially, in every branch of the ``good'' subspace, exactly one $s_{i} = 1$. Consequently, the valid Dicke state $\ket{D^n_k}$ from block $T_{i}$ is swapped into $Q$. Since every good branch results in $Q$ containing the \emph{same} state $\ket{D^n_k}$, the target register factors out of the superposition:
\begin{equation}
\ket{\varphi_\text{good}} \propto \left( \sum_{\x \neq 0^M} c_{\x} \ket{\text{junk}(\x)}_{T,A,S} \right) \otimes \ket{D^n_k}_Q.
\end{equation}
The result is a tensor product where the output $Q$ is disentangled from the junk registers.

\paragraph{4) Constant-Depth Amplification.}The final state is a mixture of the target $\ket{D^n_k}$, with probability $\gamma \approx 1 - e^{-p_\text{good}\cdot M}$, and the failure state $\ket{0^n}$, with probability $1-\gamma$. By choosing $M > 1/p_\text{good} \approx \sqrt{n}$, $\gamma$ becomes a constant. We can thus apply a constant number of rounds of exact amplitude amplification to boost the success probability to exactly 1, achieving the desired exact implementation of the Dicke state.
\else 
\paragraph{1) Initialization.}
We initialize a register $T$ (of size $m \cdot n$) into $m$ smaller registers (blocks) $T_1, \dots, T_m$, each of size $n$. On each block $T_i$, we apply independent rotations (as in \Cref{sec:reduction_overview}) to create a biased state
\begin{align} \label{eqn:veta_decomp}
    \ket{\veta_\theta}_{T_i} = \alpha \ket{D^n_k}_{T_i} + \beta \ket{\perp}_{T_i},
\end{align}
where $|\alpha|^2 = o(1)$, $|\alpha|^2 + |\beta|^2 = 1$, and $\ket{\perp}$ is a state orthogonal $\ket{D^n_k}$. Across the full register $T$, the overall state is in the superposition
\begin{align}
     \ket{\veta_\theta}^{\otimes m}_T = \sum_{\x \in \{0,1\}^m} \left( \prod_{i \in [m]} \alpha^{x_i} \beta^{
        1-x_i} \right) \bigotimes_{i \in [m]} \left(\ket{D^n_k}^{x_i} \ket{\perp}^{1-x_i} \right)_{T_i}.
\end{align}
Note that in this superposition, each branch has an associated $m$-bit string $\x$, which tracks the blocks that are in the Dicke-subspace, i.e. in the state $\ket{D^n_k}$.

\paragraph{2) Amplifying to a ``Block-$W$'' State.} Across the superposition, we will now coherently isolate the subset of branches containing exactly one register $T_i$ in the state $\ket{D^n_k}$ (i.e. $|\x|=1$). The high-level procedure is visualized in \Cref{fig:qac0f_amplify}

\begin{figure}[t]
    \centering
    \includegraphics[width=0.65\linewidth]{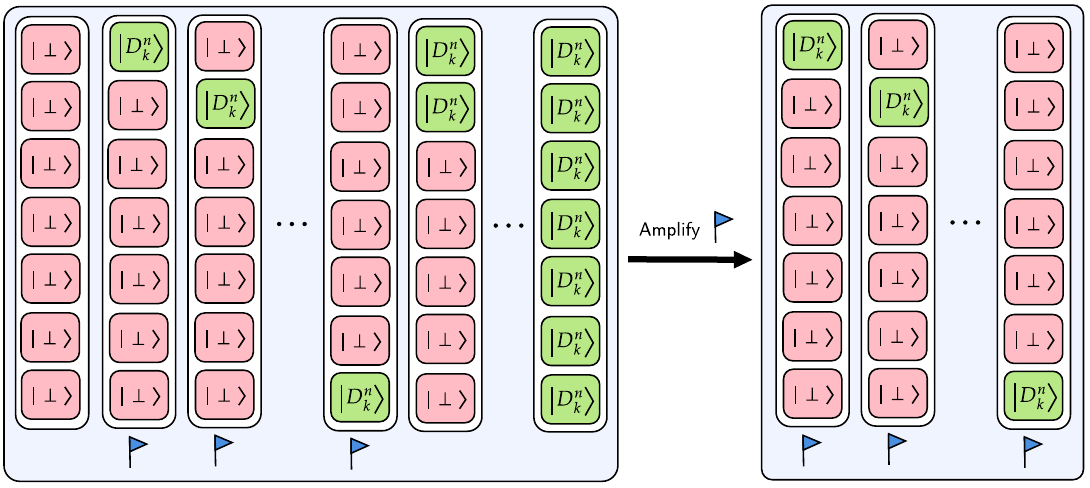}
    \caption{\textbf{Amplifying to a ``Block-$W$'' State.} In the left blue box we abstractly depict the branches of the superposition of the state $\ket{\veta_\theta}^{\otimes m}_T$, in terms of its $\ket{D^n_k}$ (green boxes) and $\ket{\perp}$ (red boxes) components within each of the $m$ registers. Using the ability to perform \EXACT$_k$ and \EXACT$_1$, an ancilla is flagged (depicted with a blue flag) for branches of the superposition containing exactly one $\ket{D^n_k}$ state (with all others in the $\ket{\perp}$ state). By choosing $m$ carefully, we can ensure that the marked states have constant probability amplitude. Thus, oblivious amplitude amplification is used to coherently map the state to a uniform superposition over only these branches with exactly one Dicke state (i.e. a ``block-W'' state), as depicted by the resultant state in the blue box on the right.}
    \label{fig:qac0f_amplify}
\end{figure}

To implement this via a \QACZf~circuit, we first associate with each register $T_i$ a ``flag'' ancilla register, denoted \textsf{flag}$(i)$. By applying the \EXACT$_k$ function to register $T_i$ and mapping its output to \textsf{flag}$(i)$, this ancilla will be set to $\ket{0}$ if $T_i$ is in the $\ket{\perp}$ state and $\ket{1}$ if in the state $\ket{D^n_k}$. We then associate one final ancilla register to the system, which we will denote as $\textsf{mark}$. By applying an \EXACT$_1$ gate to all $\textsf{flag}(i)$ ancillae and mapping the output $\textsf{mark}$, this final ancillae tracks whether $|\x|=1$ and, thus, the superposition branch contains exactly one Dicke register. Overall, the system is mapped to the state
\begin{align}
    \sum_{\x \in \{0,1\}^m} \left( \prod_{i \in [m]} \alpha^{x_i} \beta^{
        1-x_i} \right) \bigotimes_{i \in [m]} \left(\left(\ket{D^n_k}^{x_i} \ket{\perp}^{1-x_i} \right)_{T_i}\ket{x_i}_{\textsf{flag}(i)}\right) \otimes \ket{\EXACT_1(\x)}_{\textsf{mark}}.
\end{align}

Let us now denote $p:=|\alpha|^2$ as the probability that a given register measures to the Dicke state.
While each of the individual registers has small probability $p$ with the Dicke state, the overall probability amplitude across all states containing exactly one Dicke state is exactly $m p (1-p)^{m-1}$. Therefore, if we set the number of copies $m=\Theta(1/p)$, this probability has a constant value. We choose $m$ large enough to allow a single round of amplitude amplification to boost the state exactly to the superposition over only branches with $|\x|=1$, as desired. Thus, the system is mapped to the state
\begin{align}
    \frac{1}{\sqrt{m}}\sum_{\substack{\x \in \{0,1\}^m:\\|\x|=1}}\:\: \bigotimes_{i \in [m]} \left(\left(\ket{D^n_k}^{x_i} \ket{\perp}^{1-x_i} \right)_{T_i}\ket{x_i}_{\textsf{flag}(i)}\right) \otimes \ket{1}_{\textsf{mark}}.
\end{align}

Finally, recall that the probability of measuring $\ket{\veta_\theta}$ to be a Dicke state was binomially distributed according to $k$. For the most challenging cases in which $k=n/2$, as expressed in \Cref{eqn:binom}, the probability is $\Omega(1/\sqrt{n})$. Thus, in the worst case, the total number of state copies needed is $m=O(\sqrt{n})$.

\paragraph{3) Output Register Initialization.} We will now introduce an $n$-qubit output register, which we will denote \textsf{out}. This register will be initialized to the state $\ket{\perp}$. 

To see why $\ket{\perp}$ can be prepared via a \QACZf circuit note that in \Cref{eqn:veta_decomp} $|\alpha|^2=o(1)$, which implies that $|\beta|^2=\Omega(1)$. In other words, the state $\ket{\veta_\theta}$ has constant overlap with the $\ket{\perp}$ state. Therefore, we can use a similar procedure to that of  \Cref{thm:exact_dicke} to amplify $\ket{\veta_\theta}$ to $\ket{\perp}$ via a \QACZf~circuit. The key difference is that instead of using \EXACT$_k$ to mark and amplify the Dicke subspace, $\neg$\EXACT$_k$ is used to mark and amplify about the perpendicular subspace.

\paragraph{4) Dicke Extraction via Controlled-SWAP.} 
Conceptually, in this step, we will swap the single Dicke state in each branch of the superposition into the output register \textsf{out}. The high-level procedure is depicted in \Cref{fig:qac0f_extract}.

\begin{figure}[t]
    \centering
    \includegraphics[width=0.55\linewidth]{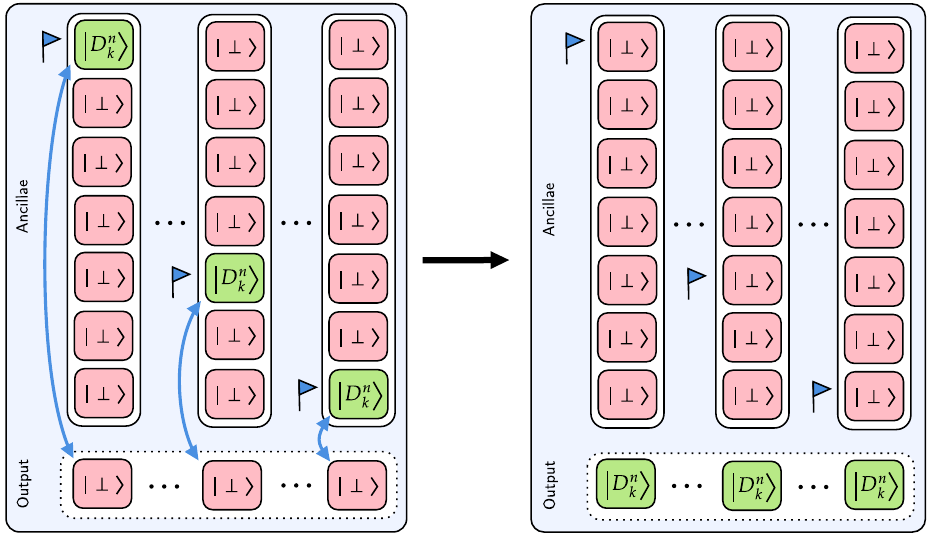}
    \caption{\textbf{Dicke Extraction via Controlled-SWAP.} An abstract visualization of the mapping from the block-$W$ state to the desired Dicke state (unentangled with junk ancillae). Specifically SWAP operations (depicted as blue arrows) are controlled by ancillae (depicted as blue flags) that flag where the Dicke state is located in each branch of the superposition. These extract the singular Dicke state from each branch of the block-$W$ state and map them into an output register. The Dicke is replaced by the $\ket{\perp}$ state that the output register was initialized to, enabling efficient uncomputation of the ancillae.}
    \label{fig:qac0f_extract}
\end{figure}

Concretely, we perform a parallel controlled-\texttt{SWAP} operation: for each $i \in [m]$, we swap the contents of block $T_i$ and target \textsf{out} conditioned on the flag ancilla bit $\textsf{flag}(i)$, i.e. $\prod_{i\in [m]} \texttt{C}_{\textsf{flag}(i)}\texttt{SWAP}(T_i, \textsf{out})$. It may seem counterintuitive that we can perform $m=O(n)$ distinct controlled-\SWAP~operations that all act on the same \textsf{out} register in a constant-depth circuit. However, in \Cref{sec:parallization}, we show this is possible to achieve in parallel due to the $W$-structure of the controls. Specifically, each controlled-\SWAP~operation can be decomposed into three $\ccnot$ gates, with $\ccnot$ gates across difference controlled-\SWAP operations commuting inside the $W$-subspace. This thus guarantees that at most one control is active at a time. 

Crucially, since the system was mapped to a block-$W$ state in Step (2), exactly one $\textsf{flag}(i^*) = 1$. Consequently, the valid Dicke state $\ket{D^n_k}$ from block $T_{i^*}$ is swapped into \textsf{out} and replaced with $\ket{\perp}$. Following this operation, each branch of the superposition contains the \emph{same} state $\ket{D^n_k}$ in \textsf{out} and, thus, the output register becomes unentangled with the other registers. Specifically, the system is mapped to the state
\begin{align} \label{eqn:intermed_state}
    \frac{1}{\sqrt{m}}\sum_{\substack{\x \in \{0,1\}^m:\\|\x|=1}}\:\: \bigotimes_{i \in [m]} \ket{\perp}_{T_i}\ket{\x_i}_{\textsf{flag}(i)} \otimes \ket{1}_{\textsf{mark}}\ket{D^n_k}_\textsf{out}.
\end{align}

\paragraph{5) Cleanup of Ancillae.} Let \textsf{flag}$=\bigcup_{i \in [m]}\textsf{flag}(i)$ denote the $m$-qubit register consisting of all the 1-qubit $\textsf{flag}(i)$ ancilla registers. By close inspection of the system state described in \Cref{eqn:intermed_state}, it can be seen that the state can be re-expressed simply as $\ket{D^n_k}_\textsf{out}\ket{\perp}^{\otimes m}_T \ket{D^m_1}_\textsf{flag} \ket{1}_\textsf{flag}$. Note, however, that all the ancillae of this state can be  uncomputed by a \QACZf~circuit. In particular, \textsf{mark} is uncomputed by a simple bit-flip, each $T_i$ is uncomputed by running the circuit to prepare $\ket{\perp}$ (as discussed in Step 3) in reverse, and \textsf{flag} is uncomputed by running the \QACZ~circuit that prepares a weight-1 $m$-qubit Dicke state in reverse.
\fi

\section{Comparison to Prior Work} \label{sec:comp_prior_work}
To contextualize our results, we compare them with recent works on low-depth Dicke state preparation. In Table \ref{tab:prior_works}, we provide a comprehensive overview of the results, emphasizing trade-offs between connectivity, circuit-depth, measurement access, and ancilla overhead.

\paragraph{The Log-Depth Unitary Barrier.} Early approaches to Dicke state preparation were limited by device geometry. As seen in the first section of Table \ref{tab:prior_works}, restrictions to 1D chains \cite{bartschi2022detdicke} or 2D grids \cite{bartschi2022shortdicke, yuan2025depth} resulted in constructions scaling quadratically or linearly with system size. Relaxing to all-to-all connectivity resulted in significantly improved logarithmic-depth scaling \cite{cruz2019w, bartschi2022shortdicke, liu2024low, yuan2025depth}. 

However, within the standard unitary circuit model, composed of bounded-width gates, these approaches face a fundamental depth lower-bound. Concretely, it is well-established in this setting that global entanglement generation, such as that of $W$ states, requires depth $\Omega(\log n)$ \cite{moore2001parallel}. Recent work of \cite{yuan2025depth} further provided a matching lower-bound of $\Omega(\log n)$ for Dicke state preparation in any ancilla-free circuit built from 2-qubit gates. Additionally, note that, in the context of Matrix Product States (MPS), Dicke states possess a bond-dimension of $k+1$ \cite{raveh2024dicke}. \cite{malz2024prep} demonstrated that while these MPS can be efficiently represented, their unitary preparation requires depth scaling logarithmically with system size to build up the necessary correlations from a product state. Similarly, \cite{gui2024space} established general space-time tradeoffs showing that for ``sparse'' states (e.g. Dicke states), surpassing logarithmic-depth scaling generally requires exponential ancillae or moving to measurement-based models. Overall, these works show that in the standard circuit model (with bounded-width gates), sub-logarithmic depth preparation of Dicke states is not possible.

\begin{table}[t!]\label{tab:prior}
    \centering
    \scriptsize
    \renewcommand{\arraystretch}{1.4}
    \setlength{\tabcolsep}{4pt}
    \begin{tabular}{|c|c|c|c|c|c|c|c|}
        \hline
        \rowcolor{gray!10} & \multicolumn{5}{c|}{Circuit} & \multicolumn{2}{c|}{Dicke State Output} \\
        \hhline{|~|-----|--|}
        \multirow{-2}{*}{\cellcolor{gray!10}Paper}&\cellcolor{gray!10}Type & \cellcolor{gray!10}Connectivity &\cellcolor{gray!10}Interactions &\cellcolor{gray!10}Depth &\cellcolor{gray!10}Ancillae &\cellcolor{gray!10}Type &\cellcolor{gray!10}Weight $(k)$ \\
         \hline
         \hline
         \cite{bartschi2022detdicke} & Unitary & 1D Chain & $\calO(1)$-Width & $\calO(n)$ & 0 & Exact & $[1,n/2]$ \\
         \hline
         \cite{bartschi2022shortdicke} & Unitary & $(n_1 \times n_2)$-Grid& $\calO(1)$-Width & $\calO(k \sqrt{\frac{n}{k}})$ & 0 & Exact & $[n_2/n_1,n/2]$ \\
         \hline
         \cite{yuan2025depth} & Unitary & $(n_1 \times n_2)$-Grid& $\calO(1)$-Width & $\calO(k \log\frac{n}{k}+n_2)$ & 0 & Exact & $[n_2/n_1,n/2]$ \\
         \hline
         \cite{yuan2025depth} & Unitary & $(n_1 \times n_2)$-Grid& $\calO(1)$-Width & $\calO(n_2)$ & 0 & Exact & $[1,n_2/n_1]$ \\
         \hline
         \hline
         \cite{cruz2019w} & Unitary & All-to-All& $\calO(1)$-Width & $\calO(\log n)$ & 0 & Exact & $1$ \\
         \hline
         \cite{bartschi2022shortdicke} & Unitary & All-to-All& $\calO(1)$-Width & $\calO(k \log \frac{n}{k})$ & 0 & Exact & $[1,n/2]$ \\
         \hline
         \cite{liu2024low} & Unitary & All-to-All& $\calO(1)$-Width & $\calO(\log^3 n \cdot  \log\log n)$ & $\calO(n \log n)$ & Exact & $[1,n/2]$ \\
         \hline
         \cite{liu2024low} & Prob. & All-to-All& $\calO(1)$-Width & $\calO(\log n \cdot  \log\log n)$ & $\calO(n \log n)$ & Exact & $[1,n/2]$ \\
         \hline
         \cite{yuan2025depth} & Unitary & All-to-All& $\calO(1)$-Width & $\calO(\log k \log \frac{n}{k}+k)$ & 0 & Exact & $[1,n/2]$ \\
         \hline
         \hline
         \cite{Buhrman2024statepreparation} & \textsf{LAQCC} & Local Grid & $\calO(1)$-Width & $\calO(1)$* & $\calO(n \log n)$ & Exact & $1$ \\
         \hline
        \cite{Buhrman2024statepreparation} & \textsf{LAQCC} & Local Grid & $\calO(1)$-Width & $\calO(1)$* & $\calO(n^2 \log n)$ & Exact & $\calO(\sqrt{n})$ \\
        \hline
        \cite{Buhrman2024statepreparation} & \textsf{LAQCC} & 1D Chain & $\calO(1)$-Width & $\calO(\log n)$* & $\Omega(n^2)$ & Exact & $[1,n/2]$ \\
        \hline
        \cite{yu2025efficient} & \textsf{LAQCC} & All-to-All& $\calO(1)$-Width & $\poly\log(n)$* & $\calO(\log n)$ & Exact & $[1,n/2]$ \\
        \hline
        \hline
         \cite{farrell2025digital} & \textsf{LOCC} & Grid & $\calO(1)$-Width & $\calO(1)$* & 0 & Approximate & $1$ \\
         \hline
         \cite{piroli2024approximate} & \textsf{LOCC} & Grid & $\calO(1)$-Width & $\calO(\log k)$* & $\calO(1)$ & Approximate & $[1,n/2]$ \\
         \hline
         \cite{piroli2024approximate} & \textsf{LOCC} & Grid& $\calO(1)$-Width & $\calO(1)$* & $\calO(\log k)$ & Approximate & $[1,n/2]$ \\
         \hline
         \hline
         \cite{grier2026qac0}\tablefootnote{Note that \cite{grier2026qac0} was concurrent to this work.} & Unitary & All-to-All& Global \CZ &  $\calO(1)$ & $\poly(n)$  & Exact & 1\\
         \hline
         \Cref{thm:exact_dicke} & \cellcolor{green!15}Unitary & All-to-All& Global \CZ & \cellcolor{green!15}$\calO(1)$ & $\calO(n^{k+1})$ & \cellcolor{green!15}Exact & \cellcolor{green!15}$\calO(1)$ \\
         \hline
         \Cref{thm:wn_qac} & \cellcolor{green!15}Unitary & All-to-All& Global \CZ & \cellcolor{green!15}$\calO(1)$ & \cellcolor{green!15}$\calO(1)$  & Approximate & 1\\
         \hline
         \Cref{thm:exact_qac0f} & \cellcolor{green!15}Unitary & All-to-All& Global \FANOUT & \cellcolor{green!15}$\calO(1)$ &  $O(n^2 \sqrt{\log n})$ & \cellcolor{green!15}Exact & \cellcolor{green!15}$[1,n/2]$ \\ 
         \hline 
    \end{tabular}
    \caption{A comprehensive overview of recent results on efficient Dicke state preparation. We group relevant results and sort them by general type: Unitary, \textsf{LAQCC} (Local Adaptive Quantum Circuits with Classical Communication), \textsf{LOCC} (Local Operations and Classical Communication), and Prob. (Probabilistic). Note that circuit models that are non-unitary, which involve non-trivial classical post-processing have a $(*)$ following the depth, to indicate additional classical circuit depth. Furthermore, all weights are listed in the range $[1,n/2]$, since Dicke states in the range $[n/2,n]$ can simply be prepared from those states and application of a layer $X$ gates.  The main advantages of our results relative to the prior literature are highlighted in green.}
    \label{tab:prior_works}
\end{table}

\paragraph{Limitations of Measurements.} Recent research has sought to surpass this logarithmic-depth barrier by moving beyond unitary evolution to models that incorporate intermediate measurements and classical feed-forward. These approaches typically fall under the framework of Local Operations and Classical Communication (\textsf{LOCC}) or its circuit-complexity formalization, Local Adaptive Quantum Circuits with Classical Communication (\textsf{LAQCC}). Unlike standard circuits which execute a fixed sequence of gates, these models allow for mid-circuit measurements where outcomes are processed classically to determine and adaptively apply future quantum gates.

While works such as \cite{Buhrman2024statepreparation} and \cite{piroli2024approximate} leverage this adaptivity to achieve constant quantum depth, they face fundamental theoretical constraints rooted in the classification of entanglement. It is well known that GHZ and Dicke states belong to distinct \textsf{SLOCC} (Stochastic \textsf{LOCC}) equivalence classes \cite{dur2000three}. While this separates their fundamental entanglement structure, more crucial to their compilation is the preparation complexity. To this end, \cite{piroli2021assist} established that while \textsf{LOCC} can deterministically prepare GHZ states in constant depth (placing them in a "trivial" phase), it cannot also do so for particle-conserving Dicke states due to their $U(1)$-symmetry. Their work proves that exact, deterministic preparation of Dicke states requires depth scaling with system size if interactions are restricted to local grids. This theoretical bound explains why the \textsf{LOCC} results in Table \ref{tab:prior_works} are necessarily probabilistic and approximate \cite{piroli2024approximate}, or require significant ancilla overheads \cite{Buhrman2024statepreparation}.

\paragraph{Introducing Global Interactions.} Our results, highlighted in the final section of \Cref{tab:prior_works}, demonstrate that the logarithmic-depth barrier can be overcome without the trade-offs inherent to \textsf{LOCC} by leveraging global interactions. Specifically, by leveraging global \CZ and \FANOUT gates, we achieve constant depth for exact Dicke state preparation. Furthermore, unlike the measurement-based approaches (constrained by \cite{piroli2021assist}), our protocols for constant-weight Dicke states in \QACZ (\Cref{thm:exact_dicke}) and arbitrary-weight Dicke states in \QACZf (\Cref{thm:exact_qac0f}) are exact and deterministic.

Finally, we acknowledge concurrent and independent work by \cite{grier2026qac0}, who similarly established constant-depth preparation of $W$ states via global \CZ gates (i.e. in \QACZ). Our work complements and extends these findings in two key directions. First, while \cite{grier2026qac0} focuses on just the $k=1$ case in \QACZ, we offer \QACZ circuits for exact preparation of any constant-weight, i.e. $k=\calO(1)$, Dicke state. We also offer \QACZf circuits for exact preparation of arbitrary-weight, i.e. $k \in [1, n/2]$, Dicke states. Second, while their $W$ state protocol requires $\poly(n)$ ancillae, we propose an additional approximate protocol (\Cref{thm:wn_qac}) requiring only $\mathcal{O}(1)$ ancillae.

\section{Discussion} \label{sec:discussion}
Despite the seemingly theoretical nature of complexity classes like \QACZ and \QACZf, the results of this work have immediate implications for physical hardware. Specifically, our protocols offer a framework for ranking the constant-depth computational power of near-term quantum architectures, based on their native connectivity. 

At the base of this hierarchy lie fixed-coupling architectures (e.g. standard superconducting circuits grids) which are fundamentally restricted by light-cone constraints, requiring $\Omega(\log n)$-depth to generate global entanglement. We demonstrate that platforms equipped with global interactions can break this barrier, but our results further suggest a nuanced distinction between the types of global control available.

Specifically, our work points to a potential power separation between neutral atom and trapped ion systems. As discussed in \Cref{sec:intro}, global \CZ gates are natural to the Rydberg blockade mechanism of neutral atom arrays, whereas global \FANOUT interactions are native to trapped ion platforms via Mølmer-Sørensen interactions. The fact that our exact protocol for arbitrary-weight Dicke states seems to rely crucially on global \FANOUT (\Cref{thm:exact_qac0f}), whereas our \CZ-based construction (\Cref{thm:exact_dicke}) appears restricted to constant-weight Dicke states, suggests that trapped ion architectures may possess a distinct state-synthesis advantage over neutral atoms for preparing highly entangled states with extensive excitations. This potential separation motivates several key directions for future inquiry:
\begin{itemize}
    \item \textbf{Rigorous Separation of Architectures:} Can it be proven that exact preparation of arbitrary-weight Dicke states is impossible in constant-depth using only global \CZ gates? Specifically, establishing a lower-bound for preparation of any $\omega(1)$-weight Dicke state in \QACZ would prove a strict computational separation between \QACZ (native to neutral atom architectures) and \QACZf (native to trapped ion architectures). This would resolve the long-standing open question of \cite{moore1999qac0} and have implications for the relative constant-depth power of different quantum hardware architectures.
    \item \textbf{Resource Optimization:} Can the resource requirements (i.e. depth and ancilla-overhead) for our exact Dicke state preparation protocols be compressed?
    \item \textbf{Alternative Global Primitives:} Do other hardware architectures, characterized by distinct forms of global connectivity, map to other interesting constant-depth complexity classes? It remains to be seen if alternative global interaction mechanisms offer superior compilers for Dicke states or other families of highly-entangled states.
\end{itemize}

\section{Acknowledgements}

The authors thank both Avishay Tal and John Wright for helpful discussions throughout this project. In particular, Tal provided key insight for the ancilla-efficient implementation of the \EXACT~Boolean function. The authors  thank Jeffery Yu for clarifications regarding prior work. F.V. also thanks  Kewen Wu for an enlightening early conversation regarding exact preparation of GHZ states in \QACZ. M.J. thanks Lucas Gretta for useful discussions about limiting fanout in classical circuits and Meghal Gupta for useful discussions about Dicke states. Finally, F.V. acknowledges ChatGPT for several helpful interactions and insights throughout the course of this project, as well as Gemini for assistance in preparation of the manuscript. Both authors are supported by the U.S. Department of Energy, Office of Science, under Award No. DE-SC0024124. F.V. is additionally supported by the Paul and Daisy Soros Fellowship for New Americans. M.J. is additionally supported by NSF grant 2311733. 
\bibliographystyle{alpha}
\bibliography{main}
\appendix
 \newpage\section{Constant-Weight \THRESH~and \EXACT~in \texorpdfstring{\QACZ}{Lg}} \label{sec:exact_one}
In this section we offer explicit \QACZ circuits for implementation of the \EXACT$_k$ Boolean function (as defined in \Cref{eqn:exact_defn}).
In particular, we show that an exact implementation of $\EXACT_k$ (for constant $k$) is achievable in \QACZ with $\calO(n^{k+1})$ ancillae (\Cref{sec:const_exact_proof}). We also show that a constant-error approximate implementation of \EXACT$_1$ is achievable in \QACZ with just $O(1)$ ancillae (\Cref{sec:approx_exact_proof}).

\subsection{Exact Constant-Weight \EXACT~in \texorpdfstring{\QACZ}{Lg}} \label{sec:const_exact_proof}
To achieve the exact implementation of \EXACT$_k$, we will first offer an explicit circuit for the $\THRESH_k$ Boolean function (\Cref{thm:threshk_qac0}) and then use \Cref{fact:exact_thresh} to strightforwardly compute \EXACT$_k$ (\Cref{thm:exactk_qac0}). For ease of notation, we will denote often denote by $\THRESH_k$ the shorthand $\THR_k$.

\begin{lemma}[Constant-Weight $\THRESH_k$ in \QACZ] \label{thm:threshk_qac0}
For any $k = O(1)$, $\textnormal{\THRESH}_k$ can be implemented in \QACZ using $O(n^{k+1})$ ancillae. 
\end{lemma}
\begin{proof}[Proof of \Cref{thm:threshk_qac0}]
We will describe a recursive procedure for implementing $\THR_k$ in terms for $\THR_{k'}$ for $k' < k$. The base case is when $k = 0$ and can be implemented with a single $\neg \OR(\x)$ gate. 

Let $\ell = \lceil \log n \rceil$. For any $j \in [n]$, let $\bn(j)_i$ be the $i$th bit of $\bn(j)$, the binary encoding of $j$. Define $\S_{i,b} \subseteq [n]$ to be the set of numbers $j \in [n]$, such that the $i^\text{th}$ bit of their binary representation is $b$. Formally, for $i \in [\ell]$ and $b \in \bin$, the set is defined as 
\begin{align}
    \S_{i,b} = \{j \in [n]:~\bn(j)_i = b\}.
\end{align}
Then, for $k > 0$, $\THR_k$ is given by the following recurrence.  
\begin{align}\label{eq:thresh}
    \THR_{k} = \andl_{i \in \ell} \lr{ \
     \THR_0(\x_{\S_{i,0}}) \lor \THR_0(\x_{\S_{i,1}}) \ 
    \orl_{k' \in [k-1]} \lr{\THR_{k'}(\x_{\S_{i,0}}) \land \THR_{k-k'}(\x_{\S_{i,1}}) }}
\end{align}
First we will argue the correctness of \Cref{eq:thresh} and then analyze its size and depth.
\paragraph{$\THR(\x) = 1$:} For any partitioning of the qubits into sets $S$ and $\ov{S}$, it must be that $|\x_{\S}| + |\x_{\ov{S}}| \leq k$. Thus, either $|\x_{\S}|=0$ or $|\x_{\ov{S}}| = 0$, or for $k' = |\x_{\S}|$, 
$\THR_{k'}(\x_{S}) \land \THR_{k'}(\x_{\ov{S}})$ is \emph{true}. Therefore, each \OR clause in \Cref{eq:thresh} will be satisfied making the entire formula \emph{true}.

\paragraph{$\THR(\x) = 0$:} Let $j_1,j_2 \in [n]$ be two \emph{different} qubits such that $x_{j_1} = 1$ and $x_{j_2} = 1$. Then, since $j_1 \neq j_2$, there must be some $i \in [\ell]$ such that $\bn(j_1)_i \neq \bn(j_2)_i$. For this value of $i$, both $\THR_0(\x_{\S_{i,0}}) = 0$ and $\THR_0(\x_{\S_{i,1}}) = 0$. Furthermore, for any $S$, $|\x_{S}| + |\x_{\ov{S}}| > k$. Thus, for any $k' \in [k]$, the clause $\THR_{k'}(\x_{\S_{i,0}}) \land \THR_{k-k'}(\x_{\S_{i,1}})$ will be \emph{false}. This makes the \OR clause corresponding to $i$ \emph{false} making the entire formula \emph{false}. 

\paragraph{Circuit analysis.}
Without loss of generality, we will assume $n$ is a power of $2$ and therefore $n = 2^\ell$. This circuit can be computed by a bottom-up dynamic program in $d(k)$ additional layers after creating $f(\ell,k)$ copies of each input qubit $x_j$ at the beginning. 
In each step, we recursively compute all the values $\THR_t(\x_{\S_{i,b}})$ for $t \in \rng{k}, i\in [\ell], b \in \bin$  on a new register $q_{t,i,b}$. For each value of $t$, this can be done in parallel across all $\S_{i,b}$. Each $\S_{i,b}$ can be done in depth $d(t)$ given $f(\ell-1,k-1)$ (from $\vlr{\S_{i,b}} = n/2$) of the copies corresponding to the qubits in $\S_{i,b}$. Finally, these values can be combined in $3$ layers on $q_{t,i,b}$ to obtain $\THR_k$ using \Cref{eq:thresh}. 
The total depth is given by the recurrence $d(k) = 3 + \sum_{t = 0}^{k-1} d(t)$, which satisfies $d(k+1) \leq 2d(k)$ giving $d(k) = O(2^k)$. 
The total number of copies $f(\ell,k)$ to be created at the start is given by the recurrence, 
\begin{align}
    f(\ell,k) &= \ell \cdot f(\ell-1,k-1) \leq \ell^k
\end{align}
because each $j \in [n]$ appears in $\ell$ sets. 
This can be done using \Cref{fact:grier_exact} with $O(n^{k+1}) = \poly(n)$ total ancillae.
\end{proof}

\begin{corollary}[Exact \EXACT$_k$ in \QACZ] \label{thm:exactk_qac0}
    For any $k = O(1)$, $\textnormal{\EXACT}_k$ can be implemented in \QACZ using $O(n^{k+1})$ ancillae. 
\end{corollary}
\begin{proof}[Proof of \Cref{thm:exactk_qac0}]
    By \Cref{fact:exact_thresh}, \EXACT$_k$ can be implemented in \QACZ with the same asynmptotic ancilla overhead as the $\THRESH_k$ implementation.
\end{proof}

\subsection{Reduced-Ancillae Approximation of \EXACT\texorpdfstring{$_1$}{Lg} in \texorpdfstring{\QACZ}{Lg}} \label{sec:approx_exact_proof}
The previous exact implementation of constant-weight \EXACT$_k$ Boolean functions in \QACZ required polynomial ancillae. We will now show how the \EXACT$_1$ unitary can be approximated to arbitrary constant precision in \QACZ~using only constant ancillae. This will ultimately be used (in \Cref{sec:qac0_dicke_proofs}) to achieve a constant-error approximation of the $W$ state in \QACZ~with constant ancillae.

\begin{lemma}[Approximate \EXACT$_1$ in \QACZ] \label{thm:approx_exact_imp}
    For any constant error $\varepsilon \in (0,1)$, there exists a \QACZ circuit $C_{\{\calS_i\}_{i\in k}}$ parameterized by a random choice of subsets $\{\calS_i\subset [n]\}_{i\in k}$, using $a=O(1)$ ancillae, such that for all $\x \in \{0,1\}^n$,
    \begin{align}
        \Pr_{\{\calS_i\}_{i\in k}} \left[C_{\{\calS_i\}_{i\in k}}\ket{\x}\ket{0^a} \neq \ket{\x}\ket{\textnormal{\EXACT}_1(\x)}\ket{0^{a-1}}\right] \leq \varepsilon.
    \end{align}
\end{lemma}
\begin{proof}[Proof of \Cref{thm:approx_exact_imp}]
    We will first describe a high-level construction of the circuit which approximates the \EXACT$_1$ unitary and justify its correctness. We will then show how all components of the circuit can be implemented in \QACZ~with constant ancillae.

    \paragraph{Circuit Proposal and Correctness.} Our goal is to implement a circuit that checks if $h_{\x}=1$. To do so, we will first construct a simplified gadget that checks if all the bits in some random subset $\calS \subset [n]$, i.e. $\x_\calS$, are of value 0 and if at least one bit in $\bar{\calS}$ has value 1. We will then re-run this gadget a constant number of times, $t$, and perform a threshold test on all the output. Via a Chernoff bound, we will show that the overall probability of error in the threshold decays exponentially in $t$.

    Concretely, our gadget is constructed as follows. For each bit in $i \in [n]$, with probability $1/2$, add the bit $i$ to subset $\calS$ or, if not, add it to the complementary subset $\bar{\calS}$. Based on this partitioning, we will split the input string $\x$ into the two sub-strings $\x_{\calS}$ and $\x_{\bar{\calS}}$. The gadget $G_\calS$ performs the following check that: $A$) all bits in $\x_\calS$ are 0 and $B$) at least one bit of $\x_{\bar{\calS}}$ is 1, as
    \begin{align}
        G_\calS(\x) = A(\x_\calS) \land B(x_{\bar{\calS}}) =\left(\bigwedge_{i \in \calS} \neg~x_i\right) \land \left(\bigvee_{j \in \bar{\calS}} x_j\right).  
    \end{align}

    We will now evaluate the probability of success for a gadget $G_\calS$ given an input string $\x \in \bin^n$, with Hamming weight $|\x|$. In the case that $|\x|=0$, the second clause $B(\x_{\calS})$ always evaluates to zero, meaning 
    \begin{align}
        \Pr_\calS\:[~G_\calS(\x)=1~|~h_{\x} = 0~] = 0.
    \end{align}
    For $|\x| \geq 1$, the probability that the first clause accepts, i.e. $A(\x_\calS)=1$, is equal to the probability that none of the 1 bit strings end up in $\calS$. Since the probability that a given bit string ends up in $\calS$ is $1/2$, 
    \begin{align}
        \Pr_\calS\:[~A(\x_\calS)=1~] = 1/2^{|\x|}.
    \end{align}
    Furthermore, in the case that clause $A(\x_{\calS})$ is satisfied and $|\x| \geq 1$, it is always true that clause $B(\x_{\bar{\calS}})$ is satisfied, which implies
    \begin{align}
        \Pr_{\calS}\:[~G_{\calS}(\x)=1~|~h_{\x} \geq 1~] = 1/2^{|\x|}.
    \end{align}
    Therefore, in the desired YES instance, in which $|\x| =1$, the probability of acceptance is $1/2$.
    Meanwhile, in the NO instance, where $|\x|=0$ or $|\x| \geq 2$, the probability of acceptance is upper-bounded by $1/4$ (it is strictly equal to $1/4$ when $|\x| =2$).

    To improve the probability of acceptance in the YES case and widen the gap with the NO case, the circuit will re-run the gadget a constant number of times, $t$, with randomly chosen $\calS_t$ for each gadget instance. Let $N^{(t)}_{|\x|}$ denote the total number of successes  when the gadget is run $t$ times on an input string with Hamming weight $|\x|$. The circuit will then use a threshold operation on these $t$ outcomes to accept if $N^{(t)}_{|\x|} > 3t/8$.

    We will now analyze the overall probability of error for the full circuit, in which $t$ instances of the gadget are run and the outcomes are thresholded. To begin, note that $N^{(t)}_{|\x|}$ is a random variable distributed as
    \begin{align}
        N^{(t)}_{|\x|} \sim \begin{cases}
            0, & \text{if } |\x| = 0 \\
            \text{Bin}(t,1/2), & \text{if } |\x| = 1 \\
            \text{Bin}(t,1/2^{|\x|}), & \text{if } |\x| \geq 2
        \end{cases}, \quad \text{s.t.} \quad \Ex [N^{(t)}_{|\x|}] = \begin{cases}
            0, & \text{if } |\x| = 0 \\
            t/2, & \text{if } |\x| = 1 \\
            t/2^{|\x|}, & \text{if } |\x| \geq 2
        \end{cases}.
    \end{align}
    In the YES instance, in which $|\x|=1$, an error occurs if $N^{(t)}_{1} \leq 3t/8$. By Chernoff's lower-tail bound, we have that
    \begin{align}
        \Pr\left(N^{(t)}_{|\x|} \leq (1-\delta)\cdot \Ex [N^{(t)}_{|\x|}]\right) \leq \exp\left(-\frac{\Ex [N^{(t)}_{|\x|}]\cdot \delta^2}{2}\right),
    \end{align}
    which implies that
    \begin{align}
        \Pr\left(N^{(t)}_{1} \leq 3t/8\right) \leq \exp\left(-t/64\right).
    \end{align}
    In the NO instance, in which $h_{\x}\neq 1$, an error occurs if $N^{(t)}_{|\x|} \geq 3t/8$. 
    For $|\x|=0$, we previously established that the probability of error is zero. For $|\x| \geq 2$, the case with the largest expected number of trials and, thus, highest probability of error is the case in which $|\x|=2$. Using Chernoff's upper-tail bound,
    \begin{align}
        \Pr\left(N^{(t)}_{|\x|} \geq (1+\delta)\cdot \Ex [N^{(t)}_{|\x|}]\right) \leq \exp\left(-\frac{\Ex [N^{(t)}_{|\x|}]\cdot \delta^2}{2+\delta}\right).
    \end{align}
    we get that
    \begin{align}
        \Pr\left(N^{(t)}_{|\x| \geq 2} \geq 3t/8\right) \leq \Pr\left(N^{(t)}_{2} \geq 3t/8\right)  \leq \exp\left(-t/40\right).
    \end{align}
    Therefore, for any string $\x$, of arbitrary Hamming weight, the probability that the circuit does not correctly compute \EXACT$_1$ is upper-bounded by $\exp\left(-t/64\right)$. This implies that if we run the gadget 
    \begin{align}
        t \geq \lceil 64\log(1/\varepsilon) \rceil
    \end{align}
    times, the circuit achieves any desired constant error $\varepsilon$.

    \paragraph{Circuit Implementation in \QACZ.} Now that we have described the high-level circuit structure and proven its correctness, we will describe how the different circuit components can be implemented in \QACZ.

    The first key component of the circuit is the gadget that checks $G_\calS(\x)$, which can be implemented as follows. 
    To check if all bits in $\x_\calS$ are 0 the gadget will apply an $X$ gate to all the qubits in $\calS$, perform a Toffoli gate controlled on $\calS$ with target acting on an ancilla initialized to 0, and then perform X gates again on $\calS$---thereby computing the value $A(\x_\calS)$ into the ancilla. To check if at least one of the bits in $\x_{\bar{\calS}}$ is 1, it suffices to check that $\x_{\bar{\calS}}$ is not the all zeros string, i.e. $B(\x_{\bar{\calS}}) = \neg A(\x_{\bar{\calS}})$. Therefore, we can simply reuse the previous gadget for computing $A$, apply it to the subset of qubits in $\bar{\calS}$ and apply an $X$ gate to the ancilla target qubit. Finally, we apply a Toffoli gate with the two ancillae as controls and as target a fresh ancilla initialized to zero, thereby computing $G_\calS(\x)$ into the ancilla. To reset the other two ancillae to zero (ensuring the computation is input preserving), we can then simply re-apply the previously discussed gates for computing $A(\x_\calS)$ and $\neg A(\x_{\bar{\calS}})$. The overall circuit requires 3 ancillae overall (although 2 are returned back to the initial 0 state and can be reused).

    The circuit will rerun the $G_\calS(\x)$ gadget a constant number of times, $t$, in which each time the subset $\calS$ is apriori randomly choosen. These $t$ runs can either be done in parallel (requiring less depth, but a linear number of ancillae) or in parallel (requiring increased constant depth, but only constant ancillae). After the gadget is run $t$ times, there will be $t$ ancillae, each containing the value $G_\calS(\x) \in \{0,1\}$ for a randomly choosen $\calS$. Therefore the $t$ ancillae are in a state $\ket{y}$, for some string $y \in \{0,1\}^t$. 
    
    The final step of the circuit is to implement a gadget which computes whether $|y| \geq 3t/8$. (For simplicity, we will assume $t$ is choosen to be a multiple of 8, such that the threshold is an integer value.) Although \THRESHOLD on $n$ qubits is generically not in \ACZ~or \QACZ, note that here we are only computing \THRESHOLD on a \emph{constant} number of bits. 
    % In particular, we can simply use our previous exact \THRESH implementation (\Cref{thm:threshk_qac0}). Note that while the ancilla count of this gate is $n^{\calO(t)}$, since here $n=t$, the overall ancilla count is $t^{\calO(t)}=\calO(1)$ in our setting where $t$ is constant.
    As such, we can decompose the \THRESHOLD into a DNF that brute-force checks all subsets of size $T \geq 3k/8$, i.e.
    \begin{align}
        \sum_i y_i \geq T \iff \bigvee_{\substack{S\subseteq [k]:\\ |S|\geq 3k/8}} \bigwedge_{i \in S} y_i.
    \end{align}
    Since there are only a constant number of subsets, overall this will only require constant ancilla overhead and depth to implement. 
    
    Therefore, by running the $G_\calS(\x)$ checks sequentially, the circuit can perform the desired computation in \QACZ~using only a constant number of ancillae.
\end{proof}

\section{Preparing Constant-Weight Dicke States in \texorpdfstring{\QACZ}{Lg}} \label{sec:qac0_dicke_proofs}
Leveraging our exact implementation of constant-weight \EXACT~and our approximate implementation of \EXACT$_1$ from the previous section, we will now show how to prepare exact constant-weight Dicke states and approximate $W$ states in \QACZ. 

Specifically, we begin by establishing a Dicke-to-\EXACT~reduction, which for any $k=\calO(1)$ uses an \EXACT$_k$ phase oracle to prepare a weight-$k$ Dicke state (\Cref{sec:exact_dicke_reduct}). We then plug the exact implementations of \EXACT$_k$, from \Cref{sec:const_exact_proof}, into this reduction to achieve exact weight-$k$ Dicke states in \QACZ (\Cref{sec:exact_qac0_proof}). Finally, we plug the approximate implementation of \EXACT$_k$, from \Cref{sec:approx_exact_proof}, into this reduction and perform error analysis to achieve a constant-error approximation of the $W$ state in \QACZ (\Cref{sec:approx_w_analysis}).

\subsection{A Constant-Weight Dicke-to-\EXACT~Reduction for \texorpdfstring{\QACZ}{Lg}} \label{sec:exact_dicke_reduct}

We will now show that, for any constant weight $k=\calO(1)$, preparation of the Dicke state $\ket{D_k^n}$ can be reduced, via amplitude amplification, to a phase oracle implementation of \EXACT$_k$ in \QACZ.

\begin{lemma}[Reducing $\ket{D_k^n}$ to \EXACT$_k$] \label{thm:w_to_exact}
    Given a \QACZ circuit $C'$ that implements the \textnormal{\EXACT}$_k$ phase oracle (stores answer in the phase), there exists a \QACZ~circuit $C$ that prepares the $\ket{D_k^n}$ state, without any additional ancillae.
\end{lemma}
\begin{proof}[Proof of \Cref{thm:w_to_exact}]
    We assume our circuit is initialized to the state $\ket{0^{n}}$. In the first layer, the circuit performs a series of single-qubit rotations, parameterized by angle $\theta$, mapping the initial state to the separable state
    \begin{align}
        \ket{\veta_\theta} = R_y(2 \theta)^{\otimes n} \ket{0^{n}} &= \left( \cos\theta \ket{0}+\sin\theta \ket{1} \right)^{\otimes n} \\
        & = \alpha_\theta \sum_{|x|=k} \ket{x} + \sqrt{1-\alpha_\theta^2} \sum_{|y| \neq k} \ket{y} \\
        &= \alpha_\theta \cdot \sqrt{{n \choose k}} \cdot  \ket{D^n_k} + \sqrt{1-\alpha_\theta^2} \sum_{|y| \neq 1} \ket{y}
    \end{align}
    where $\alpha_\theta = (\sin\theta)^k \cdot (\cos\theta)^{n-k}$.
    We will hereon denote that overlap between this state and the desired Dicke state $\ket{D^n_k}$ as
    \begin{align}
        p^n_k(\theta) = |\braket{\veta_\theta | {D^n_k}}|^2 = |\alpha_\theta|^2 \cdot {n \choose k} = {n \choose k} \cdot (\sin^2\theta)^k \cdot (\cos^2\theta)^{n-k}.
    \end{align}

    We will now use the Intermediate Value Theorem to argue that, for any number of qubits $n$ and constant $c_k \in (0,e^{-k})$, there exists a choice of angle $\theta \in [0,\pi/2]$ such that $p^n_k(\theta)=c_k$. Begin by noting that $p^n_k(0)=0$ and $p^n_k(\pi/2)=0$. Furthermore, as a function of $\lambda = \sin^2\theta \in [0,1]$, 
    \begin{align}
        p^n_k(\lambda) = {n \choose k} \: \lambda^k \: (1-\lambda)^{n-k}
    \end{align}
    is the Binomial($k,n,\lambda$) probability of exactly $k$ success. This function is maximized by $\lambda=k/n$ and, for any $n \geq k \geq 1$, can be lower-bounded by
    \begin{align}
        p^n_k\left(\frac{k}{n}\right) = {n \choose k} \cdot \left(\frac{k}{n}\right)^k \cdot \left(1-\frac{k}{n}\right)^{n-k} \geq \left(\frac{n}{k}\right)^k \cdot \left(\frac{k}{n}\right)^k \cdot \left(1-\frac{k}{n}\right)^{n-k} = \left(1-\frac{k}{n}\right)^{n-k}.
    \end{align}
    Furthermore, leveraging the fact that for any $x \in [0,1]$, $\ln(1-x) \geq -\frac{x}{1-x}$, for any $k<n$ we have that
    \begin{align}
        \ln \left[\left(1-\frac{k}{n}\right)^{n-k}\right] &= (n-k) \cdot \ln \left(1-\frac{k}{n}\right) \geq  (n-k) \cdot \left(-\frac{k/n}{1-k/n}\right) \\
        &= (n-k) \cdot \left(-\frac{n \cdot k/n}{n-k}\right) = -k,
    \end{align}
    which implies that 
    \begin{align}
        p^n_k\left(\frac{k}{n}\right) \geq \left(1-\frac{k}{n}\right)^{n-k} \geq e^{-k},
    \end{align}
    where $e^{-k}=\calO(1)$, since $k$ is constant. Therefore, the function $p_n(\theta)$ has value 0 at $\theta=0$, increases until it reaches maximal value $\geq e^{-k}$ at $\theta=1/n$, and decreases back to 0 at $\theta=\pi/2$. Since $p^n_k(\theta)$ is continuous, by the Intermediate Value Theorem there exists at least two choices of $\theta^* \in [0,\pi/2]$ such that $p^n_k(\theta^*)=c_k \in (0,e^{-k})$.

    With this, we can now leverage amplitude amplification to show that there exists a \QACZ circuit which amplifies this success probability to exactly 1. To begin, we specify the two reflection oracles that comprise the so-called ``Grover oracle'' for amplitude amplification. The first reflection operation is the reflection about the ``initial'' state, $\ket{\veta_\theta}$, defined as
    \begin{align}
        U_\text{init} = I - 2\ketbra{\veta_\theta}{\veta_\theta} = I- 2 R_y(2 \theta)^{\otimes n} \ketbra{0^{n}}{0^{n}}(R_y(2 \theta)^{\otimes n})^\dagger.
    \end{align}
    Note that since the state $\ket{\veta_\theta}$ is itself a product state, this gate is a depth-1 \QACZ operation.
    The second reflection is a reflection about the ``winner'' state, which applies a $-1$ phase to the $\ket{D^n_k}$ state. In \QACZ, this reflection can be achieved via simple application of the \EXACT$_k$ phase oracle, i.e.
    \begin{align}
        U_\text{win} = C'.
    \end{align}
    Thus, the overall Grover oracle is defined as
    \begin{align}
        G = U_\text{init} U_\text{win} = (I - 2\ketbra{\veta_\theta}{\veta_\theta}) \cdot C'.
    \end{align}
    In running amplitude amplification, we will simply apply the Grover oracle $\ell$ times to the initial state $\ket{\veta_\theta}$, thereby obtaining the state $\ket{\psi_\ell}=G^\ell \ket{\veta_\theta}$. All that remains to show is that there exists some constant $\ell=O(1)$ such that $\ket{\psi_\ell}=\ket{D^n_k}$.

    Each Grover iteration of the amplitude amplification procedure will rotate the success probability by angle $\omega_\theta$, where
    \begin{align}
        \sin^2\omega_\theta = p^n_k(\theta).
    \end{align}
    This implies that after $\ell$ iterations the ``success'' probability is
    \begin{align}
        P_\text{succ}(\theta, \ell) = \sin^2\left((2\ell+1)\cdot \omega_\theta\right).
    \end{align}
    Therefore, in order to achieve $P_\text{succ}(\theta, \ell)=1$, we need
    \begin{align}
        (2\ell+1)\cdot \omega_\theta = \frac{\pi}{2} \:\:  (\bmod \pi) \iff  p^n_k(\theta) =\sin^2\left(\frac{\pi}{4\ell+2}\right).
    \end{align}
    By the previous Intermediate Value Theorem-based argument, we saw that we can pick $\theta$ such that $p^n_k(\theta) = c_k$ for any constant $c_k \in (0,e^{-k})$. Therefore, we simply choose the smallest integer $\ell$ such that 
    \begin{align}
        c_k = \sin^2\left(\frac{\pi}{4\ell+2}\right) < e^{-k}.
    \end{align}
    Thus,
    \begin{align}
        \ell = \left\lceil\frac{1}{4} \left(\frac{\pi}{\sqrt{\arcsin(e^{-k})}}-2\right) \right\rceil,
    \end{align}
    which is constant valued since $k=\calO(1)$. Therefore, there exists a choice of $\theta$ such that our amplitude amplification procedure maps the amplitude of the desired $\ket{D^n_k}$ state to exactly $1$, in a constant number of rounds $\ell$. Letting $d$ denote the depth of the \QACZ circuit $C'$ implementing the \EXACT$_k$ phase oracle, the overall depth of this \QACZ circuit is, thus, $(d+1)^\ell$.
\end{proof}

\subsection{Exact Preparation of \texorpdfstring{$\ket{D^n_k}$}{Lg} in \texorpdfstring{\QACZ}{Lg}} \label{sec:exact_qac0_proof}
We will now show how our exact implementation of constant-weight \EXACT~functions can be leveraged in the Dicke-to-\EXACT~reduction to achieve exact \QACZ circuits for constant-weight Dicke state preparation.

\exactdickeqaczero*

\begin{proof}[Proof of \Cref{thm:exact_dicke}]
     \Cref{thm:exactk_qac0} offers an explicit \QACZ circuit for computing the \EXACT$_k$ Boolean function into an output qubit, for any $k=\calO(1)$. Meanwhile, the Dicke-to-\EXACT~reduction, of \Cref{thm:w_to_exact}, shows that an \EXACT$_k$ \emph{phase oracle} is sufficient to implement a weight-$k$ Dicke state in \QACZ. Therefore, all that remains is to show that the \EXACT$_k$ circuit can be converted into a phase oracle. Note that such a phase oracle can be implemented by simply running the \EXACT$_k$ circuit, applying a $Z$ gate to the output, and running the \EXACT$_k$ circuit again in reverse---to uncompute the output and only leave the phase.
\end{proof}
\noindent We further note that this result straightforwardly implies (via the lower-bounds of \cite{parham2025quantumcircuitlowerbounds}) that \QACZ is not contained in the first level of the magic hierarchy.
\qaczmagic*
\begin{proof}[Proof of \Cref{thm:qacz_magic}]
     \cite{parham2025quantumcircuitlowerbounds} showed that there exists a polynomial $p$ such that \textsf{MH$_1$} cannot prepare states that are $p(n)$ close to the $W$ state. However, by \Cref{thm:exact_dicke}, \QACZ can exactly prepare the $W$ state.
\end{proof}

\subsection{Approximate Computation of \texorpdfstring{$\ket{W_n}$}{Lg} in \texorpdfstring{\QACZ}{Lg}} \label{sec:approx_w_analysis}
We will now show how our approximate implementation of the \EXACT$_1$ function in \QACZ can be leveraged in the Dicke-to-\EXACT~reduction to achieve a constant-error fidelity approximation of the $W$ state in \QACZ, using only constant ancillae.

\approxw*
\begin{proof}[Proof of \Cref{thm:wn_qac}]
We will show how to choose the approximation quantity in \Cref{thm:approx_exact_imp}
so that the resulting \QACZ~circuit  has fidelity at least $1-\varepsilon$
with the ideal $\ket{W_n}$ state.

\paragraph{Approximate \EXACT$_1$ on Basis States.}
Let $\textbf{S}=\{\calS_i\}_{i \in [t]}$ denote an instantiation of the random subsets for the gadget described in \Cref{thm:approx_exact_imp}.  This lemma provides, for any constant error parameter
$\eta = O(1)$, a \QACZ~circuit $\widetilde U_{\textbf{S}}$, instantiated by the choice of $\textbf{S}$, requiring $O(1)$ ancillae
that approximately computes \EXACT$_1$ into an output bit. 
Specifically, for each input $\x \in \{0,1\}^n$, let
\begin{align}
    \delta_{\x}^{\textbf{S}}= \mathbf{1}\left\{\widetilde U_{\textbf{S}}\ket{\x}\ket{0^a} \neq \ket{\x}\ket{\EXACT_1(\x)}\ket{0^{a-1}}\right\}
\end{align}
be the indicator that the circuit $\widetilde U_{\textbf{S}}$, instantiated with $G_{\calS_i}$ gadgets for the selected subsets $\calS_i \in \textbf{S}$, incorrectly computes $\EXACT_1(\x)$. Thus, the action of the unitary on input state $\ket{\x}$ is 
\begin{equation}
  \widetilde U_{\textbf{S}}\ket{\x}\ket{0^a} =
  (1-\delta_{\x}^{\textbf{S}})\ket{\x}\ket{E_1(\x)}\ket{0^{a-1}}
  +  \delta_{\x}^{\textbf{S}}\ket{\x}\ket{E_1(\x)\oplus 1}\ket{0^{a-1}}.
  \label{eq:UE1-pointwise}
\end{equation} 
Furthermore, by \Cref{thm:approx_exact_imp}, note that, for any constant error $\eta$ and all $\x \in \bin^n$,
\begin{align} \label{eqn:exp_delta}
    \Ex_{\textbf{S} \in \mathfrak{S}}\left[\delta_{\x}^{\textbf{S}}\right] = \Pr_{\textbf{S} \in \mathfrak{S}} \left[U_{\textbf{S}}\ket{\x}\ket{0^a} \neq \ket{\x}\ket{\EXACT_1(\x)}\ket{0^{a-1}}\right]\leq \eta.
\end{align}

\paragraph{Approximate \EXACT$_1$ on the Initial Superposition.}
Let $U_E$ denote the unitary that exactly computes the \EXACT$_1$ function into the output qubit.
Denote the ``ideal'' initial state $\ket{\psi}$ of the exact $\ket{W_n}$-state reduction (before the amplitude amplification step) from the \EXACT-to-Dicke reduction (Lemma~\ref{thm:w_to_exact}) as
\begin{equation}
  \ket{\psi}
  \;:=\;
  U_{E} R^{\otimes n}\ket{0^n}\ket{0^a}
  \;=\;
  \sum_{\x \in \{0,1\}^n} \beta_{\x}(\theta)\,\ket{\x}\ket{\EXACT_1(\x)}\ket{0^{a-1}},
\end{equation}
and the ``real'' state achieved via our approximate \EXACT$_1$ \QACZ unitary construction $\widetilde U_{\textbf{S}}$ as
\begin{align}
  \ket{\widetilde{\psi}_{\mathbf{S}}}&:=\widetilde U_{\textbf{S}} R^{\otimes n}\ket{0^n}\ket{0^a}\\
  &= \sum_{\x \in \{0,1\}^n} \beta_{\x}(\theta)
     \left(
       (1-\delta_{\x}^{\mathbf{S}})\ket{\x}\ket{\EXACT_1(\x)}\ket{0^{a-1}}
  +  \delta_{\x}^{\mathbf{S}}\ket{\x}\ket{\EXACT_1(\x)\oplus 1}\ket{0^{a-1}}
     \right).
\end{align}
The difference between these two states is the subset of inputs $\x$ for which $\widetilde U_{\textbf{S}}$ misclassifies $x$, i.e.
\begin{equation}
  \ket{\widetilde{\psi}_{\mathbf{S}}} - \ket{\psi}
  =
  \sum_{\x \in \{0,1\}^n} \beta_{\x}(\theta)
     \cdot \delta_{\x}^{\mathbf{S}} \ket{\x}\ket{\EXACT_1(\x)\oplus 1}\ket{0^{a-1}}.
\end{equation}
By triangle inequality, the norm of the error is
\begin{equation}
  \left\| \ket{\widetilde{\psi}_{\mathbf{S}}} - \ket{\psi} \right\|_2^2
  =
  \sum_{\x \in \{0,1\}^n} |\beta_{\x}(\theta)|^2
     \cdot \delta_{\x}^{\mathbf{S}},
\end{equation}
which, by \Cref{eqn:exp_delta}, has expected value
\begin{align}
    \Ex_{\textbf{S} \in \mathfrak{S}}\left[\left\| \ket{\widetilde{\psi}_{\mathbf{S}}} - \ket{\psi} \right\|_2^2\right] = \sum_{\x \in \{0,1\}^n} |\beta_{\x}(\theta)|^2
     \cdot \Ex_{\textbf{S} \in \mathfrak{S}}\left[\delta_{\x}^{\mathbf{S}}\right] \leq \eta \sum_{\x \in \{0,1\}^n} |\beta_{\x}(\theta)|^2 = \eta.
\end{align}
Since this is the expectation over all possible random choices of $\mathbf{S}$, there must exist a specific choice $\mathbf{S}^*$ such that
\begin{align} \label{eqn:err_bound}
    \left\| \ket{\widetilde{\psi}_{\mathbf{S}^*}} - \ket{\psi} \right\|_2^2
  \leq \eta.
\end{align}
It is this choice of $\mathbf{S}^*$ that we will use to instantiate the approximate \EXACT$_1$ unitary, i.e. $\widetilde U_{\textbf{S}^*}$. Furthermore, by the standard relation between Euclidean 2-norm and fidelity,
\begin{equation}
  \calF\bigl(\ket{\widetilde{\psi}_{\mathbf{S}^*}},\,\ket{\psi}\bigr)
  \ge 1 - \eta.
\end{equation}

\paragraph{Error Propagation in OAA.}
\Cref{thm:w_to_exact} showed that for a constant number of iterations $t$, there exists a choice of $\theta$ such that OAA boosts the superposition amplitude on the $W_n$ state to exactly 1. Note that for Dicke states with $k=1$ it suffices to set $t=1$. 

Let $U_W$ denote the full circuit unitary in the exact case, in which queries are made to the exact \EXACTONE unitaries $U_{E[1]}$ and $U_{E[1]}^\dagger$. Let $\widetilde{U}_W$ denote the approximate circuit, in which queries are instead made to approximate \EXACTONE unitaries $\widetilde{U}_{E[1]}^{\{S_i^*\}}$ and $\widetilde{U}_{E[1]}^{\{S_i^*\}}$. Since the circuit only uses one iteration of OAA, overall it will make two queries to $\widetilde{U}_{E[1]}^{\{S_i^*\}}$ and one query to $\left(\widetilde{U}_{E[1]}^{\{S_i^*\}}\right)^\dagger$.
By \Cref{eqn:err_bound}, each approximation perturbs the state by at most $\sqrt{\eta}$ (in Euclidean 2-norm), so a hybrid argument gives
\begin{align}
  \bigl\| \widetilde U_{\mathrm{W}}\ket{0^{n+a}} - U_{\mathrm{W}}\ket{0^{n+a}} \bigr\|_2
  \;\le\; 3 \sqrt{\eta}.
\end{align}
By the standard relation between Euclidean 2-norm and fidelity,
\begin{align}
    F\bigl( \widetilde U_{\mathrm{W}}\ket{0^{n+a}},\, \ket{W_n}\ket{0^{a-1}} \bigr) = F\bigl( \widetilde U_{\mathrm{W}}\ket{0^{n+a}},\, U_{\mathrm{W}}\ket{0^{n+a}} \bigr)
  \;\ge\; 1 - 9 \eta.
\end{align}
Therefore, if we set $\eta = \varepsilon/9$, we obtain the desired bound.
\end{proof}

\section{Parallelization Procedures Using Quantum \FANOUT} \label{sec:parallization}
We will now prove some intermediary results, showing that certain operations can be implemented in \QACZf. These will be used as subroutines in \Cref{sec:qaczf_arb_dicke}, for our proof that \QACZf can compute arbitrary-weight Dicke states.

We will first show that if there exists a \QACZf circuit to implement unitary $C$, then there exists a \QACZf circuit implementing the controlled-$C$ unitary. 
\begin{fact}[Controlled \QACZf Circuits are in \QACZf]\label{fact:controlCkt}
Let $C$ be any depth-$d$ \QACZf circuit using $m = \poly(n)$ ancillae. The following transformation is in \QACZf: 
\begin{align}
    \ket{0}_x \ket{0^{n+m}} &\mapsto \ket{0}_x \ket{0^{n+m}} \\
    \ket{1}_x \ket{0^{n+m}} &\mapsto \ket{1}_x C \ket{0^{n+m}} 
\end{align} 
\end{fact}
\begin{proof}
Note that \QACZf allows for any $\poly(n)$-\FANOUT in constant depth. Note that $C$ can have at most $m$ gates per layer. Therefore, we apply $m$-\FANOUT to create $n$ classical copies of $x$, and then run $C$ by controlling each gate of on a different copy register. Finally we can un-compute the copies by reversing the fanout.

To see why such a circuit is valid, recall that \QACZf can be defined using only reflection gates and \PARITY gates \cite{moore1999qac0} (alternatively to \FANOUT gates). Any controlled reflection gate is a valid reflection gate, and any controlled \PARITY gate can be implemented cleanly by first computing \PARITY in a separate register $t$, and then apply a $\ccnot$ gate with the control register and $t$ onto the output register, and finally uncomputing the $t$ register. 
\end{proof}

As a consequence of \cite{moore2001parallel}, we will now show that $n$ different $\ccnot$ gates, which all share either a single control or target qubit, are implementable via a \QACZf circuit. Naively one would expect this sequence of operations to require depth $n$, but we show that they can be parallelized into constant-depth via \QACZf.
\begin{claim} [Consequence of \cite{moore2001parallel}]\label{cl:fanoutpar}
Let $A = \clr{a_1, \dots a_n}$, $B = \clr{b_1, b_2 \dots b_n}$ and $t$ be $2n+1$ qubits.  Then, the unitary transformation, 
  $$\ket{\psi}_{A,B,t} \mapsto_{U_1}  \lr{ \ccnot(a_1,t,b_1) \cdot \ccnot(a_2,t,b_2) \cdots \ccnot(a_n,t,b_n)} \ket{\psi}_{A,t}$$
  as well as, 
    $$\ket{\psi}_{A,B,t} \mapsto_{U_2}  \lr{ \ccnot(a_1,b_1,t) \cdot \ccnot(a_2,b_2, t) \cdots \ccnot(a_n,b_n,t)} \ket{\psi}$$
  can both be implement in $O(1)$ depth \QACZf.
\end{claim}
\begin{proof}
Note that in the case when the $\ccnot$ are replaced simply by a $\cnot$, this follows trivially from the definition of \QACZf, since the first operation would be a \FANOUT gate and the second operation would be the \PARITY gate \cite{moore2001parallel}. 

We point out that this also holds in the presence of additional controls, using the same change of basis relation between \FANOUT and \PARITY. 
The first operation, $U_1$, can be implemented by making $n$ classical copies of $t$ using the $\FANOUT$ gate, and then un-computing these copies. 
Then, observe that the second operation, $U_2$, is exactly the same as the $U_1$ after conjugating with Hadamards, i.e,
\begin{align}
 H^{\tens {n+1}}_{B,t} U_1  H^{\tens {n+1}}_{B,t}  &= 
 H_t \cdot \lr{\prod_{i \in [n]} H_{b_i} \ccnot(a_i,t,b_i) H_{b_i}} \cdot H_t \\
 &= \prod_{i \in [n]} H_{b_i} \tens H_t \cdot \ccnot(a_i,t,b_i) \cdot H_{b_i} \tens  H_t \\
 &= \prod_{i \in [n]} \ccnot(a_i,b_i,t) \\
 &= U_2
\end{align}
Therefore, $U_2$ can also be implemented in $O(1)$ depth using the $U_1$ circuit. 
\end{proof}

Finally, using our ability to parallelize the multiple $\ccnot$ gates on a shared target/control in \QACZf (i.e. \Cref{cl:fanoutpar}), we will now show that the same can be done for controlled-\SWAP~(i.e $\cswap$) gates which share a target.
\begin{corollary}\label{cor:parswap}
Let $A = \clr{a_1, \dots a_n}$, $B = \clr{b_1, b_2 \dots b_n}$ and $t$ be $2n+1$ qubits. Consider the following transformation defined on any $\y \in \bin^n$ of $|y| \leq 1$.  
  $$\ket{\y}_A \ket{\psi}_{B,t} \mapsto  \lr{ \cswap(a_1,b_1,t) \cdot \cswap(a_2,b_2, t) \dots \cswap(a_n,b_n,t)}  \ket{\y}_A \ket{\psi}_{B,t}$$
  This transformation can be implemented in $O(1)$ depth \QACZf.
\end{corollary}
\begin{proof}
   Each $\cswap(a_j, b_j,t)$ gate can be written as $\ccnot(a_j, b_j, t) \cdot \ccnot(a_j, t, b_j) \cdot \ccnot(a_j, b_j, t)$. Note for any $\y$ with $|y| \leq 1$, and $i \neq j$,
   \begin{align}\label{eq:comm}
       [\ccnot(a_j, b_j, t), \ccnot(a_i, t, b_i)]  \cdot \ket{\y}_A \ket{\psi}_{B,t} = 0
   \end{align}
    Define, 
    \begin{align}
        U_1 &:= \cswap(a_1, b_1, t) \cswap(a_2, b_2, t) \dots \cswap(a_n, b_n,t) \\
        U_2 &:= \cswap(a_1, t, b_1) \cswap(a_2, t, b_2) \dots \cswap(a_n, t, b_n) \\
    \end{align}
    Then, due to \Cref{eq:comm}, on any such $\y$, 
    \begin{align}
        \lr{ \cswap(a_1,b_1,t) \cdot \cswap(a_2,b_2, t) \dots \cswap(a_n,b_n,t)} \ket{\y}_A \ket{\psi}_{B,t} &= U_1 U_2 U_1  \ket{\y}_A \ket{\psi}_{B,t}
    \end{align}
    Then, by \Cref{cl:fanoutpar}, we can implement each of these in \QACZf in $O(1)$ depth.
\end{proof}
\section{Preparing Arbitrary-Weight Dicke States in \texorpdfstring{\QACZf}{Lg}} \label{sec:qaczf_arb_dicke}
We will now show how Dicke states $\ket{D^n_k}$ of arbitrary-weight $k \in [n]$ can be exactly and cleanly prepared in \QACZf. 

To begin, we leverage the work of \cite{takahashi2012collapse}, to straightforwardly show that \EXACT$_k$ is exactly implementable in \QACZf, for arbitrary $k \in [n]$. 
\begin{fact}[Exact Arbitrary-Weight \EXACT~in \QACZf] \label{thm:exactk_qac0f}
    For any $k \in [n]$, there exists a \QACZf circuit which computes \textnormal{\EXACT}$_k$ using $O(n \sqrt{n \log n})$ ancillae.
\end{fact}
\begin{proof}[Proof of \Cref{thm:exactk_qac0f}]
    By \Cref{fact:thresh_qac0f}, for any $k \in [n]$, $\THRESH_k$ is implementable in \QACZf. By \Cref{fact:exact_thresh}, $\EXACT_k$ can be computed from $\THRESH_k$ in \QACZf, with negligible ancilla-overhead.
\end{proof}
Additionally, we use the following amplitude amplification primitive. 

\begin{fact}[Removing Error from Marked State]\label{fact:singleamp}
Suppose the state $\ket{\psi}$ of the below form can be cleanly synthesized in \QACZf, 
$$\ket{\psi} = \sin \theta \ket{\varphi_0}_Q \ket{0}_a + \cos\theta \ket{\varphi_1}_Q \ket{1}$$
such that $\cos^2 \theta  \geq 1/4$, then, there exists a \QACZf using one additional ancilla to cleanly synthesize the state $\ket{\varphi_1}_Q$.
\end{fact}
\begin{proof}
This is essentially special case of the amplitude amplification procedure from \cite{grier2024threshold}, however to directly apply exact amplitude amplification, $\theta$ needs to be an integer multiple of $\pi/3$. Instead, we can adjust the amplitudes slightly using ancillas and achieve $\ket{\varphi_1}$ in a single round as follows. 
First, prepare an extra ancilla $b$ in the state $\ket{\nu}_b = \gamma \ket{1} + \sqrt{1-\gamma^2} \ket{0}$ where $\gamma = 1/(2 \cos \theta)$. Now the resulting state looks like, 

\begin{align}
\ket{\psi_1} &=  \frac{\sqrt{3}}{2} \ket{\varphi'}_{Q,a,b} - \frac{1}{2} \ket{\varphi_1}_Q \ket{11}_{a,b}
\end{align}
For the remaining state $\ket{\varphi'}$ in the $(I-\kb{11}_{ab})$ subspace. 
Apply the reflection gate $(I - 2\kb{11}_{ab})$ to obtain. 
\begin{align}
\ket{\psi_1} &=  \frac{\sqrt{3}}{2} \ket{\varphi'}_{Q,a,b} - \frac{1}{2} \ket{\varphi_1}_Q \ket{11}_{a,b}
\end{align}
Then, apply $(I - 2\kb{\psi_1})$ which can be implemented by $C_1^\dag (I - 2\kb{\vec{0}}_{Q,a,b}) C_1$, where $C_1$ is the circuit used to create $\ket{\psi_1}$ which in turn uses the circuit for $\ket{\psi}$. This produces, 
\begin{align}
\ket{\psi_2} &=   (I - 2\kb{\psi_1})  \frac{\sqrt{3}}{2} \ket{\varphi'}_{Q,a,b}  - \frac{1}{2} \ket{\varphi_1}_Q \ket{11}_{ab} \\
&= - \ket{\varphi_1} \ket{11}_{ab}\\
\end{align}
Finally apply a $Z$ gate on $a$ and $X \tens X$ on $a,b$ to clean up and obtain, 
\begin{align}
\ket{\psi_3} &= \ket{\varphi_1}_Q \ket{00}_{ab}
\end{align}
\end{proof}
We will now use the ability to compute arbitrary-weight $\EXACT_k$ and some ideas similar to those used in our \QACZ Dicke-to-\EXACT~reduction (\Cref{thm:w_to_exact}) to describe an explicit circuit for exact preparation of $\omega(1)$-weight Dicke states in \QACZf. 
%\malvika{todo fix the below paragraph. The proof is now simplified}
%The proof is divided into two key parts. In the first part (\Cref{thm:approx_construct}), we propose a parallelized form of amplitude amplification across $M$ copies of the initial state---each with non-constant overlap $p$ with the desired Dicke state---to obtain a state with $1-(1-p)^M$ overlap with the desired Dicke state. In the second and final part (\Cref{thm:exact_qac0f}), we observe that (with access to polynomial space in \QACZf) there exists a choice of $M = \poly(n)$ such that the overlap between the previously described state and the desired Dicke state is constant, i.e. $1-(1-p)^M=\calO(1)$. Thus, we use to constant-depth amplitude amplification to map this approximate state to the exact desired Dicke state.
%

\begin{lemma}\label{lem:cleandickef}
For any $k = \omega(1)$, the $n$-qubit weight-$k$ Dicke state $\ket{D^n_k} \ket{0^a}$ can be synthesized cleanly by a \QACZf circuit using $a=O(n^2 \sqrt{\log n})$ ancillae.  
\end{lemma}
\begin{proof}
    Similar to the proof of \Cref{thm:w_to_exact}, we begin by considering the state of the form 
    \begin{align}
        \ket{\veta_\theta} = R_y(2 \theta)^{\otimes n} \ket{0^{n}} &= \left( \cos\theta \ket{0}+\sin\theta \ket{1} \right)^{\otimes n} \\
        & = \alpha_\theta \sum_{|\x|=k} \ket{\x} + \sqrt{1-\alpha_\theta^2} \sum_{|\y| \neq k} \ket{\y} \\
        &= \alpha \ket{D^n_k} + \beta \ket{\perp}
    \end{align}
    where $\alpha_\theta = (\sin\theta)^k \cdot (\cos\theta)^{n-k}$, $\ket{\perp}$ is some state perpendicular to $\ket{D^n_k}$, and $|\alpha|^2+|\beta|^2=1$. Once again, we will define the overlap between this state and the weight-$k$ Dicke state as
    \begin{align}
        p^n_k(\theta) = |\braket{\veta_\theta | {D^n_k}}|^2 = |\alpha|^2.
    \end{align}

    Note that if there existed some $\theta$ for which $p^n_k(\theta)=O(1)$, we could simply apply the same procedure as in \Cref{thm:w_to_exact}, leveraging the fact that in \QACZf~we can implement \EXACT$_k$ for any weight $k \in [1,n-1]$, to amplify the state to the desired Dicke state in constant depth. However, as previously mentioned, we should not expect this to be the case (even for maximizing parameter $\theta=k/n$), and from hereon will in fact assume that $p^n_k(\theta)  =o(1)$. As previously discussed $p^n_k(\theta)$ is binomially distributed, such that in the worst case, i.e. when $k=n/2$, $p^n_k(\theta) = \Omega(1/\sqrt{n})$. This implies that for any $k$ we can choose a $\theta$ such that $1/p^n_k(\theta)=O(\sqrt{n})$.

    We will now argue that for this parameter regime, it is possible to synthesize the state $\ket{\perp}$ via a \QACZf circuit that using two additional ancillae. We will use this gadget later on in the algorithm. First, apply the \EXACT$_k$ function to $\ket{\veta_\theta}$, mapping the output to the ancilla register and flip the output. This produces the state $\alpha \ket{D^n_k}\ket{0} + \beta \ket{\perp}\ket{1}$. Since $p^n_k(\theta)  =o(1)$, this implies that $\alpha = o(1)$ which in turn means that $\beta^2 \gg 1/4$. Therefore, a direct application of \Cref{fact:singleamp} produces a circuit $C_{\perp}$ to synthesize $\ket{\perp}$ cleanly.  
    With these preliminary observations, we can now describe the algorithm for preparing the weight-$k$ Dicke state.

    In the following procedure, we will only explicitly track the ancilla used by our descriptions. The sub-modules such as $\EXACT_k$ and the fanout used in controlled-$\SWAP$ perform their own cleanup of the ancillas. 
    \paragraph{1) Initialization.} We begin in the all zeros state and use a layer of single-qubit rotation gates to prepare $m$ copies of the $n$-qubit state $\ket{\veta_\theta}$. For each $i \in [m]$, denote the $i^{th}$ $n$-qubit register as $T_i$. This therefore maps the system to the state
    \begin{align}
        \ket{\psi_1} & = \ket{\veta_\theta}^{\otimes m} \ket{0^{n+m}} \\ 
        &= \sum_{\x \in \{0,1\}^m} \left( \prod_{i \in [m]} \alpha^{x_i} \beta^{
        1-x_i} \right) \bigotimes_{i \in [m]} \Big(x_i \ket{D^n_k} + (1-x_i) \ket{\perp}\Big)_{T_i} \otimes \ket{0^{n+m}}.
    \end{align}
    In the second equality, the state across all $m$ copies is decomposed into a superposition over branches in the $\ket{D^n_k}$ and $\ket{\perp}$ basis. In particular, the $m$-bit string $\x \in \{0,1\}^m$ encodes the number of $\ket{D^n_k}$ versus $\ket{\perp}$ states within each state copy of the superposition branch, with $|\x|$ tracking the total number of $\ket{D^n_k}$ states.

    \paragraph{2) Flagging Dicke States.} We will now attribute to the $i^{th}$ state copy a flag ancilla register, denoted ``\textsf{flag}$(i)$'' initialized to the state $\ket{0}$. We will apply the \EXACT$_k$ function to each of the copies and map the output to the corresponding ancilla register. For each $i \in [m]$, \textsf{flag}$(i)$ will flag, in branch of the superposition, whether register $T_i$ contains the desired $\ket{D^n_k}$ state or undesired $\ket{\perp}$ state. This, therefore, maps the system to the state
    \begin{align}
        \ket{\psi_2} = \sum_{\x \in \{0,1\}^m} \left(\prod_{i \in [m]} \alpha^{x_i} \beta^{
        1-x_i}\right) \lr{\bigotimes_{i \in [m]} \Big(x_i \ket{D^n_k} + (1-x_i) \ket{\perp}\Big)_{T_i} \ket{x_i}_{\textsf{flag}(i)}} \otimes \ket{0^{n}}.
    \end{align}

    \paragraph{3) Amplifying to a ``Block-$W$'' State.}
    For simplicity, let us now denote $p:= p^n_k(\theta)$. If we pick the number of copies to be $m = \Theta(1/p)$ (which is at most $O(n)$ by the earlier discussion), then the probability mass over the superposition of all $m$ branches with \emph{exactly} one copy being flagged can be made an arbitrarily large constant. In particular, we choose $m$ to make, $p^{*} = m p (1-p)^{m-1} \geq 1/4$. 
    
    We will now assign one of the unused ancilla registers, denoted ``\textsf{mark}'', to track whether the a given branch of the superposition contains exactly one Dicke state (i.e. has exactly one register \textsf{flag}$(i)$ set to 1). In particular, denote $\textsf{flag} = \{\textsf{flag}(1), \textsf{flag}(2), \cdots, \textsf{flag}(m)\}$. We will thus apply our \QACZ~circuit for $\EXACT_1$ to the $m$-bit register \textsf{flag} and map its output into 1-qubit register \textsf{mark}, obtaining the state:
    \begin{align}
        \ket{\psi_3} = \sqrt{\frac{p^*}{m}}  &\sum_{\substack{\x \in \{0,1\}^m:\\|\x|=1}} \:\:\left(\bigotimes_{i \in [m]} \Big(x_i \ket{D^n_k} + (1-x_i) \ket{\perp}\Big)_{T_i}\ket{x_i}_{\textsf{flag}(i)} \right) \otimes  \ket{1}_\textsf{mark} \ket{0^{n-1}}\\
        &+ \sqrt{1-p^*} \ket{\text{junk}} \ket{0}_\textsf{mark} \ket{0^{n-1}}
    \end{align}
    
    Since $p^* \geq 1/4$, we can leverage \Cref{fact:singleamp} to amplify the system exactly into the portion of the superposition with \textsf{mark} set to 1. In other words, we amplify the system into a uniform superposition over all the branches containing exactly one Dicke state, i.e. the ``Block-$W$'' state:
    \begin{align}
        \ket{\psi_4} = \frac{1}{\sqrt{m}}\sum_{\substack{\x \in \{0,1\}^m:\\|\x|=1}} \:\:\left(\bigotimes_{i \in [m]} \Big(x_i \ket{D^n_k} + (1-x_i) \ket{\perp}\Big)_{T_i}\ket{x_i}_{\textsf{flag}(i)} \right) \otimes  \ket{0^{n}}.
    \end{align}
    This step is abstractly depicted in \Cref{fig:qac0f_amplify}. Since the \textsf{mark} register is already uncomputed due to \Cref{fact:singleamp}, we drop its label, allowing us to reuse the ancilla.
    
    \paragraph{4) Initializing the Output Register.} We will now use the previously described ability to compute the $\ket{\perp}$ state to prepare an $n$-qubit output register in the remaining ancillas, denoted ``\textsf{out}'', to the state $\ket{\perp}$. This maps the overall system to state
    \begin{align} \label{eqn:int_state}
        \ket{\psi_5} = \frac{1}{\sqrt{m}}\sum_{\substack{\x \in \{0,1\}^m:\\|\x|=1}} \:\:\left(\bigotimes_{i \in [m]} \Big(x_i \ket{D^n_k} + (1-x_i) \ket{\perp}\Big)_{T_i}\ket{x_i}_{\textsf{flag}(i)}\right) \otimes \ket{\perp}_{\textsf{out}}.
    \end{align}
    
    \paragraph{5) Extracting Dicke to the Output.} For each $i \in [m]$, we will apply a controlled-\SWAP$_i$ gate which, controlled on register $\textsf{flag}(i)$ being in the state $\ket{1}$, performs a \SWAP~operation between the $n$-qubit states in registers $T_i$ and $\textsf{out}$. While it might seem like the application of all $m$ of these controlled-\SWAP~operations the same \textsf{out} register would cause a depth blow-up prohibiting implementation in \QACZf, we prove in \Cref{cor:parswap} that in fact this series of controlled-\SWAP~operations can be parallelized to a constant-depth circuit. 
    
    Since each branch of the superposition in \Cref{eqn:int_state} has at most one register $T_{i^*}$ that is set to the $\ket{D^n_k}$ state (with all others set to the $\ket{\perp}$ state), only the $i^*$th flag register, i.e. $\textsf{flag}(i^*)$, will be set to the state $\ket{1}$ (with all others set to the state $\ket{0}$). Therefore, for each $i\in[m]$, the controlled-\SWAP$_i$~operation will swap the singular $\ket{D^n_k}_{T_i}$ state in each branch of the superposition with the $\ket{\perp}_{\textsf{out}}$ state in the output register. This will map the system to the state
    \begin{align} \label{eqn:int_state}
        \ket{\psi_6} = \frac{1}{\sqrt{m}}\sum_{\substack{\x \in \{0,1\}^m:\\|\x|=1}} \:\:\left(\bigotimes_{i \in [m]} \ket{\perp}_{T_i}\ket{x_i}_{\textsf{flag}(i)}\right) \otimes \ket{D^n_k}_{\textsf{out}}.
    \end{align}
    
    \paragraph{6) Uncomputing Ancillae.} Let $T= T_1 \cup T_2 \cup \cdots \cup T_m$. By shuffling the qubit ordering and observing once again that only one \textsf{flag} ancilla register was set to state $\ket{1}$ per branch of the superposition, the overall state can be re-expressed as
    \begin{align}
        \ket{\psi_6} &= \ket{D^n_k}_{\textsf{out}} \otimes \ket{\perp}^{\otimes m}_T \otimes \left(\frac{1}{\sqrt{m}}\sum_{\substack{\x \in \{0,1\}^m:\\|\x|=1}} \:\:\bigotimes_{i \in [m]}\ket{x_i}_{\textsf{flag}(i)}\right)  \\
        &= \ket{D^n_k}_{\textsf{out}} \otimes \ket{\perp}^{\otimes m}_T \otimes \ket{D^m_1}_\textsf{flag}.
    \end{align}
    
    We will now show that all the ancillary registers can be uncomputed.  Leveraging the previously discussed fact that there exists a \QACZf~circuit $C_\perp$ to compute the state $\ket{\perp}$, we can run the circuit in reverse, i.e. $C_\perp^\dagger$, to uncompute each  $\ket{\perp}$ in the $T$ register. Finally, since we previously showed that there exists a \QACZ~circuit $C$ for computing weight-1 Dicke states, $C^\dagger$ can uncompute the $\ket{D^m_1}$ Dicke state in the \textsf{flag} register. Thus, the system is mapped to the desired state $\ket{D^n_k}_{\textsf{out}} \otimes \ket{0^{n\cdot m + m}}$.

    \paragraph{Overall Ancilla Count.} We will now calculate the asymptotic ancilla-overhead of the procedure. Note that the \EXACT$_1$ operation is applied to the $m$ \textsf{flag}$(i)$ registers to  mark all the superposition branches containing exactly one Dicke state (in the \textsf{mark} register). By \Cref{thm:exact_dicke}, this has ancilla overhead $O(n^2)$. However, the overall ancilla cost is dominated by the parallel application of \EXACT$_k$ operation to each of the $m$ ancillae registers $T_i$, so as to flag each ancillae \textsf{flag}$(i)$. By \Cref{thm:exactk_qac0f}, the ancilla overhead of each  \EXACT$_k$ is $O(n \sqrt{n \log n})$. Since for $k\approx n/2$ we have that $m=O(\sqrt{n})$, this implies that the ancilla-overhead of the procedure is $O(n^2 \sqrt{\log n})$. Note that implementing the controlled-$\SWAP$ operations require fanout on the controls, which adds another $n \cdot m$ overhead, but does not change the asymptotic count. 
\end{proof}

This provides the required pieces to complete our main theorem regarding preparation of arbitrary-weight Dicke states in \QACZf. 

\exactqaczf*
\begin{proof}
When $k = O(1)$, this follows from \Cref{thm:exact_dicke} and when $k = \omega(1)$ it follows from  \Cref{lem:cleandickef}.
\end{proof}

\end{document}